\documentclass{aptpub}

\usepackage{graphicx,float,latexsym,times}
\usepackage{amsfonts,amstext,amsmath,amssymb}
\usepackage{mathrsfs}
\usepackage{color}
\usepackage{bbm}
\usepackage{euscript}       
\usepackage{mathabx}
\usepackage{booktabs} 
\usepackage{longtable}

\newcommand{\p}{\mathsf{p}}
\newcommand{\q}{\mathsf{q}}
\newcommand{\Psucc}{\Pro_{\mathsf{succ}}}

\newcommand{\T}{\EuScript{T}}
\newcommand{\Tf}{\EuScript{T}_{\rm f}}
\newcommand{\A}{\EuScript{A}}
\newcommand{\B}{\EuScript{B}}
\newcommand{\D}{\EuScript{D}}
\newcommand{\V}{\EuScript{V}}
\newcommand{\U}{\EuScript{U}}
\newcommand{\R}{\EuScript{R}}
\renewcommand{\P}{\EuScript{P}}

\newcommand{\cZ}{\mathscr{Z}}
\newcommand{\cB}{\mathscr{B}}
\newcommand{\cBk}{\cB_{t_1}^{t_k}}

\newcommand{\I}{\mathbbm{I}}
\newcommand{\ind}[1]{\I_{\{#1\}}}
\newcommand{\1}{\mathbbm{1}}
\newcommand{\Pro}{\mathbb{P}}
\newcommand{\Exp}{\mathbb{E}}

\newcommand{\qed}{{\hfill\ensuremath{\square}}}

\authornames{MOUSTAKIDES, LIU, MILENKOVIC} 
\shorttitle{The secretary problem with random queries} 

\begin{document}

\title{Optimal Stopping Methodology\\ for the Secretary Problem with Random Queries} 

\authorone[University of Patras]{George V. Moustakides} 
\authortwo[Xi'an Jiaotong-Liverpool University]{Xujun Liu} 
\authorthree[University of Illinois, Urbana-Champaign]{Olgica Milenkovic} 
\addressone{Department of Electrical and Computer Engineering, University of Patras, Rion, GREECE} 
\emailone{moustaki@upatras.gr} 
\addresstwo{Department of Foundational Mathematics, Xi’an Jiaotong-Liverpool University, Suzhou, CHINA}
\emailtwo{Xujun.Liu@xjtlu.edu.cn}
\addressthree{Department of Electrical and Computer Engineering, University of Illinois, Urbana-Champaign, IL, USA} 
\emailthree{milenkov@illinois.edu} 



\begin{abstract}
Candidates arrive sequentially for an interview process which results in them being ranked relative to their predecessors. Based on the ranks available at each time, one must develop a decision mechanism that selects or dismisses the current candidate in an effort to maximize the chance to select the best. This classical version of the ``Secretary problem'' has been studied in depth using mostly combinatorial approaches, along with numerous other variants. In this work we consider a particular new version where during reviewing one is allowed to query an external expert to improve the probability of making the correct decision. Unlike existing formulations, we consider experts that are not necessarily infallible and may provide suggestions that can be faulty. For the solution of our problem we adopt a probabilistic methodology and view the querying times as consecutive stopping times which we optimize with the help of optimal stopping theory. For each querying time we must also design a mechanism to decide whether we should terminate the search at the querying time or not. This decision is straightforward under the usual assumption of infallible experts but, when experts are faulty, it has a far more intricate structure. 
\end{abstract}



\section{Introduction}

The secretary problem, also known as the game of Googol, the beauty contest problem, and the dowry problem, was formally introduced in \cite{G}, while the first solution was obtained in \cite{L}. A widely used version of the secretary problem can be stated as follows: $n$ individuals are ordered without ties according to their qualifications. They apply for a ``secretary'' position, and are interviewed one by one, in a uniformly random order. When the $t$-th candidate appears, she or he can only be ranked (again without ties) with respect to the $t-1$ already interviewed individuals. At the time of the $t$-th interview, the employer must make the decision to hire the person present or continue with the interview process by rejecting the candidate; rejected candidates cannot be revisited at a later time. The employer succeeds only if the best candidate is hired. If only one selection is to be made, the question of interest is to determine a strategy (i.e., final rule) that maximizes the probability of selecting the best candidate. 

In \cite{L} the problem was solved using algebraic methods with backward recursions while \cite{D} considered the process as a Markov chain. An extension of the secretary problem, known as the \emph{dowry problem with multiple choices} (for simplicity, we refer to it as the dowry problem), was studied in \cite{GM}. In the dowry problem, one is allowed to select $s>1$ candidates during the interview process, and the condition for success is that the selected group includes the best candidate. This review process can be motivated in many different ways: for example, one may view the $s$-selection to represent candidates invited for a second round of interviews. In \cite{GM} we find a heuristic solution to the dowry problem while \cite{S1} solved the problem using a functional-equation approach of the dynamic programming method. 

In \cite{GM} the authors also offer various extensions to the secretary problem in many different directions. For example, they examine the secretary problem (single choice) when the objective is to maximize the probability of selecting the best or the second-best candidate, while the more generalized version of selecting one of the top $\ell$ candidates was considered in \cite{GZ}. Additionally in \cite{GM} the authors studied the full information game where the interviewer is allowed to observe the actual values of the candidates, which are chosen independently from a known distribution. Several other extensions have been considered in the literature: the postdoc problem, for which the objective is to find the second-best candidate \cite{GM,R};
the problem of selecting all or one of the top $\ell$ candidates when $\ell$ choices are allowed \cite{N,T1,T2}. 
More recently, the problem is considered under a model where interviews are performed according to a nonuniform distribution such as the Mallows or Ewens, \cite{CJMTUW,J,LM,LMM}. For more details regarding the history of the secretary problem, the interested reader may refer to \cite{F1,F2}.

The secretary problem \textit{with expert queries}, an extension of the secretary problem with multiple choices, was introduced in~\cite{LMM} and solved using combinatorial methods for a generalized collection of distributions that includes the Mallows distribution. In this extended version it is assumed that the decision making entity has access to a limited number of infallible experts. When faced with a candidate in the interviewing process an expert, if queried, provides a binary answer of the form ``The best'' or ``Not the best''. Queries to experts are frequently used in interviews as an additional source of information and the feedback is usually valuable but it does not mean that it is necessarily accurate. This motivates the investigation of the secretary problem when the response of the expert is not deterministic (infallible) but it may also be false. This can be modeled by assuming that the response of the expert is \textit{random}. Actually, the response does not even have to be binary as in the random query model employed in \cite{CPM,MS} for the completely different problem of clustering. In our analysis we consider more than two possibilities which could reflect the level of confidence of the expert in its knowledge about the current candidate being the best or not. For example a four-level response could be of the form: ``The best with high confidence'', ``The best with low confidence'', ``Not the best with low confidence'' and ``Not the best with high confidence''. As we will see, the analysis of the problem under a randomized expert model requires stochastic optimization techniques and in particular results we are going to borrow from optimal stopping theory \cite{PS,S2}. 

The idea of augmenting the classical information of relative ranks with auxiliary \textit{random} information (e.g.~coming from a fallible expert) has also been addressed in \cite{DLLV}. In this work the authors consider various stochastic models for the auxiliary information which is assumed to become available to the decision maker \textit{with every new candidate}. The goal is the same as in the classical secretary problem, namely optimize the final termination time. 
This must be compared to the problem we are considering in our current work where auxiliary information becomes available \textit{only} at querying times which constitute a \textit{sequence} of stopping times that must be selected optimally. Furthermore, at each querying time, using the extra information provided by the expert, we also need to decide, optimally, whether we should terminate the selection process at the querying time or continue to the next querying. Our problem formulation involves three different stochastic optimizations (i.e.~sequence of querying times, decision whether to stop or continue at each querying time, final termination time) while the formulation in \cite{DLLV} requires only the single optimization of the final termination time. We would like to emphasize that the optimization of the sequence of querying times and the optimization of the corresponding decision to stop or continue at each querying time is by no means a simple task. It necessitates a careful analysis with original mathematical methodology, constituting the main contribution of our work.

Finally, in \cite{AGKK} classical information is augmented with machine learned advice. The goal is to assure an \textit{asymptotic} performance guarantee of the value maximization version of the secretary problem (where one is interested in the actual value of the selection and not the order). As in the previous reference, there are no queries to an expert present and as we pointed out the analysis is asymptotic with no \textit{exact} (non-asymptotic) optimality results as in our work.

\section{Background}\label{sec:2}
We begin by formally introducing the problem of interest along with our notation.
Suppose the set $\{\xi_1,\ldots,\xi_n\}$ contains objects that are random uniform draws without replacement from the set of integers $\{1,\ldots,n\}$. The sequence $\{\xi_t\}_{t=1}^n$ becomes available sequentially and \textit{we are interested in identifying the object with value equal to 1}, which is regarded as ``the best''. The difficulty of the problem stems from the fact that the value $\xi_t$ is not observable. Instead, at each time $t$, we observe the relative rank $z_t$ of the object $\xi_t$ after it is compared to all the previous objects $\{\xi_1,\ldots,\xi_{t-1}\}$. If $z_t=m$ (where $1\leq m\leq t$) this signifies that in the set $\{\xi_1,\ldots,\xi_{t-1}\}$ there are $m-1$ objects with values strictly smaller than $\xi_t$. As mentioned, at each time $t$ we are interested in deciding between $\{\xi_t=1\}$ and $\{\xi_t>1\}$ based on the information provided by the relative ranks $\{z_1,\ldots,z_t\}$.

Consider now the existence of an expert we may query. The purpose of querying at any time $t$ is to obtain from the expert the information about the \textit{current} object being the best or not. Unlike all articles in the literature that treat the case of deterministic expert responses here, as mentioned in the Introduction, we adopt a \textit{random} response model. In the deterministic case the expert provides the exact answer to the question of interest and we obviously terminate the search if the answer is ``$\{\xi_t=1\}$''. In our approach the corresponding response is assumed to be a \textit{random} number $\zeta_t$ that can take $M$ different values. The reason we allow more than two values is to model the possibility of different levels of confidence in the expert response. Without loss of generality we may assume that $\zeta_t\in\{1,\ldots,M\}$ and the probabilistic model we adopt is the following
\begin{equation}
\Pro(\zeta_t=m|\xi_t=1)=\p(m),~~\Pro(\zeta_t=m|\xi_t>1)=\q(m),~m=1,\ldots,M,
\label{eq:rand-model}
\end{equation}
where $\sum_{m=1}^M\p(m)=\sum_{m=1}^M\q(m)=1$, to assure that the expert responds with probability 1. 
These probabilities are considered \textit{prior information known to us} and will aid us in making optimal use of the expert responses.
As we can see, the probability of the expert generating a specific response depends on whether the true object value is 1 or not. Additionally, we assume that $\zeta_t$ 
\textit{is statistically independent of any other past or future responses, relative ranks and object values}
and, as we can see from our model, only depends on $\xi_t$ being equal or greater than 1.
It is clear that the random model is more general than its deterministic counterpart. Indeed we can emulate the deterministic case by simply selecting $M=2$ and $\p(1)=1,\p(2)=0,\q(1)=0,\q(2)=1$, with $\zeta_t=1$ corresponding to ``$\{\xi_t=1\}$'' and $\zeta_t=2$ to ``$\{\xi_t>1\}$'' with probability 1.

In the case of deterministic responses it is evident that we gain no extra information by querying the expert more than once per object (the expert simply repeats the same response). Motivated by this observation we adopt the same principle for the random response model as well, namely, \textit{we allow at most one query per object}. Of course, we must point out that under the random response model, querying multiple times for the same object makes perfect sense since the corresponding responses may be different. However, as mentioned, we do not allow this possibility in our current analysis. We discuss this point further in Remark\,\ref{rem:5} at the end of Section\,\ref{sec:3}.

We study the case where we have available a maximal number $K$ of queries. This means that for the selection process we need to define the querying times $\T_1,\ldots,\T_K$ with $\T_K>\T_{K-1}>\cdots>\T_1$ (inequalities are strict since we are allowed to query at most once per object) and a final time $\Tf$ where we necessarily terminate the search. It is understood that when $\Tf$ occurs, if there are any remaining queries, we simply discard them. As we pointed out, in the classical case of an infallible expert, when the expert informs that the current object is the best we terminate the search while in the opposite case we continue with the next object. Under the random response model stopping at a querying time or continuing the search requires a more sophisticated decision mechanism. For this reason, with each querying time $\T_k$ we associate a \textit{decision function} $\D_{\T_k}\in\{0,1\}$ where $\D_{\T_k}=1$ means that we terminate the search at $\T_k$ while $\D_{\T_k}=0$ that we continue the search beyond $\T_k$. 
Let us now summarize our components: The search strategy is comprised of the querying times $\T_1,\ldots,\T_K$, the final time $\Tf$ and the decision functions $\D_{\T_1},\ldots,\D_{\T_K}$, which need to be properly optimized. Before starting our analysis let us make the following important remarks. 

\begin{remark}\label{rem:1}
It makes no sense to query or final-stop the search at any time $t$ if we do not observe $z_t=1$. Indeed, since our goal is to capture the object $\xi_t=1$, if this object occurs at $t$ then it forces the corresponding relative rank $z_t$ to become 1. 
\end{remark}

\begin{remark}\label{rem:2}
If we have queried at times $t_k>\cdots>t_1$ and there are still queries available (i.e.~$k<K$), then we have the following possibilities: 1)~Terminate the search at $t_k$;  2)~Make another query after $t_k$; and 3)~Terminate the search after $t_k$ without making any additional queries. Regarding case 3) we can immediately dismiss it from the possible choices. Indeed, if we decide to terminate at some point $t>t_k$, then it is understandable that \textit{our overall performance will not change} if at $t$ we first make a query, \textit{ignore} the expert response and then terminate our search. Of course, if we decide to use the expert response \textit{optimally} then we cannot perform worse than terminating at $t$ without querying. Hence, if we make the $k$-th query at $t_k$ it is preferable to obtain the expert response $\zeta_{t_k}$ and use it to decide whether we should terminate at $t_k$ or make another query after $t_k$. Of course if $k=K$, namely, if we have exhausted all queries, then we decide between terminating at $t_K$ and employing the final time $\Tf$ to terminate after $t_K$. We thus conclude that $\Tf>\T_K>\cdots>\T_1$.
\end{remark}

\begin{remark}\label{rem:3}
Based on the previous remarks we may now specify the information each search component is related to. Denote by $\cZ_t=\sigma\{z_1,\ldots,z_t\}$ the sigma-algebra generated by the relative ranks available at time $t$ and let $\cZ_0$ be the trivial sigma-algebra.
We then have that the querying time $\T_1$ is a $\{\cZ_t\}_{t=0}^n$-adapted stopping time where $\{\cZ_t\}_{t=0}^n$ denotes the filtration generated by the sequence of the corresponding sigma-algebras. Essentially, this means that the events $\{\T_1=t\},\{\T_1>t\},\{\T_1\leq t\}$ belong to $\cZ_t$. More generally, suppose we fix $\T_{k}=t_k,\ldots,\T_1=t_1,\T_0=t_0=0$ and for $t>t_k$ we define $\cZ_t^k=\sigma\{z_1,\ldots,z_t,\zeta_{t_1},\ldots,\zeta_{t_k}\}$ with $\cZ_t^0=\cZ_t$ then, the querying time $\T_{k+1}$ is a $\{\cZ_t^k\}_{t=t_k+1}^n$-adapted stopping time where  $\{\cZ_t^k\}_{t=t_k+1}^n$ denotes the corresponding filtration. Indeed we can see that the event $\{\T_{k+1}=t\}$ depends on the relative ranks $\cZ_t$ but also on the expert responses $\{\zeta_{t_1},\ldots,\zeta_{t_k}\}$ available at time $t$. If we apply this definition for $k=K$ then $\T_{K+1}$ simply denotes the final time $\Tf$. With the first $k$ querying times fixed as before, we can also see that the decision function $\D_{t_k}$ is measurable with respect to $\cZ_{t_k}^k$ (and, therefore, also with respect to $\cZ_t^k$ for any $t\geq t_k$). This is true because at $t_k$ in order to decide whether to stop or continue the search we use all of the information available at time $t_k$ which consists of the relative ranks and the existing expert responses (including, as pointed out, $\zeta_{t_k}$).
\end{remark}

We begin our analysis by presenting certain basic probabilities. They are listed in the following lemma.

\begin{lemma}\label{lem:1}
For $n\geq t\geq t_1>0$ we have
\begin{align}
&\Pro(\xi_t,\ldots,\xi_1)=\frac{(n-t)!}{n!},~~~\Pro(\cZ_t)=\frac{1}{t!},~~~\Pro(z_t|\cZ_{t-1})=\frac{1}{t}\label{eq:lem1-1}\\
&\Pro(\xi_{t_1}=1,\cZ_t)=\frac{1}{(t-1)!n}\ind{z_{t_1}=1}\1_{{t_1}+1}^t,\label{eq:lem1-3}\\
&\Pro(\xi_t=1,\cZ_t)=\frac{1}{(t-1)!n}\ind{z_t=1},\label{eq:lem1-2}
\end{align}
where $\mathbbm{I}_A$ denotes the indicator function of the event $A$ and where for $b\geq a$ we define $\1_a^b=\prod_{\ell=a}^b\ind{z_\ell>1}$, while for $b<a$ we let $\1_a^b=1$.
\end{lemma}
\begin{proof}
The first equality is well known and corresponds to the probability of selecting uniformly $t$ values from the set $\{1,\ldots,n\}$ without replacement. The second and third equality in \eqref{eq:lem1-1} show that the ranks $\{z_t\}$ are independent and each $z_t$ is uniformly distributed in the set $\{1,\ldots,t\}$. Because the event $\{\xi_{t_1}=1\}$ forces the corresponding rank $z_{t_1}$ to become 1 and all subsequent ranks to be greater than 1, this fact is captured in \eqref{eq:lem1-3} by the product of the indicators $\ind{z_{t_1}=1}\1_{{t_1}+1}^t$. The equality in \eqref{eq:lem1-2} computes the probability of the event $\{\xi_t=1\}$ in combination with the rank values observed up to time $t$. The details of the proof can be found in the Appendix.
\end{proof}

In the next lemma, we present a collection of more advanced equalities as compared to the ones appeared in Lemma\,\ref{lem:1} where we also include expert responses. In these identities we will encounter the event $\{\cZ_t^k,z_{t_k}=1\}$ that has the following meaning: When $t\geq t_k$ then $z_{t_k}$ is part of $\cZ_t$ which in turn is part of $\cZ_t^k$. By including explicitly the event $\{z_{t_k}=1\}$ we simply state that we fix $z_{t_k}$ to 1, while the remaining variables comprising $\cZ_t$ or $\cZ_t^k$ are free to assume any value consistent with the constraints imposed on the relative ranks.
\begin{lemma}\label{lem:2}
For $n\geq t>t_k>\cdots>t_1>0$ we have
\begin{align} \allowdisplaybreaks
&\Pro(\xi_t=1|\cZ_t^k)=\frac{t}{n}\ind{z_t=1}\label{eq:lem2-1} \allowdisplaybreaks\\  \allowdisplaybreaks
&\Pro(\xi_{t_k}=1|\cZ_{t_k}^k)=\frac{\p(\zeta_{t_k})t_k}{\p(\zeta_{t_k})t_k+\q(\zeta_{t_k})(n-t_k)}\ind{z_{t_k}=1}
\label{eq:lem2-2} \allowdisplaybreaks\\ \allowdisplaybreaks
&\Pro(\zeta_{t_k}|\cZ_{t_k}^{k-1},z_{t_k}=1)=\p(\zeta_{t_k})\frac{t_k}{n}+\q(\zeta_{t_k})\Big(1-\frac{t_k}{n}\Big) \label{eq:lem2-3}  \allowdisplaybreaks\\
 &\Pro(z_t=1|\cZ_{t-1}^k,z_{t_k}=1)
=\frac{1}{t}\Big\{1-\frac{\big(\p(\zeta_{t_k})-\q(\zeta_{t_k})\big)(t-1)}{\p(\zeta_{t_k})(t-1)+(n-t+1)\q(\zeta_{t_k})}\1_{t_k+1}^{t-1}\Big\}.\label{eq:lem2-4}\allowdisplaybreaks
 \end{align}
\end{lemma}
\begin{proof}
Equality \eqref{eq:lem2-1} expresses the fact that the probability of interest depends only on the current rank while it is independent of previous ranks and expert responses. Equalities \eqref{eq:lem2-2}, \eqref{eq:lem2-3}, \eqref{eq:lem2-4} suggest that the corresponding probabilities are functions of only the most recent expert response and do not depend on the previous responses. In particular in \eqref{eq:lem2-4} we note the dependency structure that exists between $z_t$ and past information which is captured by the indicator $\1_{t_k+1}^{t-1}$. As we argue in the proof of Lemma\,\ref{lem:1}, this indicator is a result of the fact that if $\xi_{t_k}=1$ then the ranks for times $t>t_k$ \textit{can no longer assume the value 1}. 
The complete proof is presented in the Appendix.
\end{proof}

As we will see in the subsequent analysis, compared to the deterministic case, a more challenging decision structure will emerge under the random response model. In the deterministic model \cite{GM,LMM}, deciding to stop at a querying time is straightforward. If the expert responds with ``$\{\xi_t=1\}$'' we stop, otherwise we continue our search. This strategy is not the optimum in the case of random responses since the expert does not necessarily provide binary responses and, more importantly, its responses can be faulty. As we are going to show, the optimal decision functions $\D_{\T_1},\ldots,\D_{\T_K}$ have a more intriguing form which depends on the values of the expert response and their corresponding probabilities of occurrence. Identifying the optimal search components will be our main task in the next section.

\section{Optimizing the Success Probability}\label{sec:3}
To simplify our presentation we make a final definition. For $t_k>t_{k-1}>\cdots>t_1>t_0=0$, we define the event
$$
\cBk=\{z_{t_k}=1,\D_{t_k}=0,\ldots,z_{t_1}=1,\D_{t_1}=0\},
$$
with $\cB_{t_1}^{t_0}=\cB_{t_1}^0$ denoting the whole sample space. From the definition we conclude
\begin{equation}
\I_{\cBk}=\prod_{\ell=1}^k\ind{z_{t_\ell}=1}\ind{\D_{t_\ell}=0}=\ind{z_{t_k}=1}\ind{\D_{t_k}=0}\I_{\cB_{t_1}^{t_{k-1}}}.
\label{eq:BandD}
\end{equation}
Basically $\cBk$ captures the event of querying at $t_k,\ldots,t_1$, after observing relative ranks equal to 1 (required by Remark\,\ref{rem:1}) while deciding not to terminate the search at any of these querying instances. It is clear from Remark\,\ref{rem:3} that for $t\geq t_k$ the indicator $\I_{\cBk}$ is measurable with respect to $\cZ_t^k$, because this property applies to each individual indicator participating in the product in \eqref{eq:BandD}. 

Consider now a collection of querying times and a final time satisfying $\Tf>\T_K>\cdots>\T_2>\T_1>\T_0=0$ and a corresponding collection of decision functions $\D_{\T_1},\ldots,\D_{\T_K}$ all conforming with Remark\,\ref{rem:3}. Denote with $\Psucc$ the success probability delivered by this combination, namely the probability to select the best object, then 
\begin{equation} 
\Psucc=\sum_{k=1}^K\Pro(\xi_{\T_k}=1,\D_{\T_k}=1,\cB_{\T_1}^{\T_{k-1}})+ 
\Pro(\xi_{\Tf}=1,\cB_{\T_1}^{\T_K}),
\label{eq:decomp0}
\end{equation} 
where, as usual, for any sequence $\{x_n\}$ we define $\sum_{k=a}^b x_n=0$ when $b<a$.
The general term in the sum expresses the probability of the event where we did not terminate at the first $k-1$ querying times (this is captured by $\cB_{\T_1}^{\T_{k-1}}$ which contains the event of all previous decisions being 0) and we decided to terminate at the $k$-th query (indicated by $\D_{\T_k}=1$). The single last term in \eqref{eq:decomp0} is the probability of the event where we did not terminate at any querying time (captured by $\cB_{\T_1}^{\T_K}$) and we make use of the final time $\Tf$ to terminate the search. Please note, that in \eqref{eq:decomp0} we did not include the events $\{z_{\Tf}=1\},\{z_{\T_k}=1\}$ although, as pointed out in Remark\,\ref{rem:1}, we query or final-stop only at points that must satisfy this property. This is because these events are implied by the events $\{\xi_{\T_k}=1\}$ and $\{\xi_{\Tf}=1\}$ respectively, since $\{\xi_{\Tf}=1\}\cap\{z_{\Tf}=1\}=\{\xi_{\Tf}=1\}$, with a similar equality being true for any querying time.
Let us now focus on the last term in \eqref{eq:decomp0} and apply the following manipulations
\begin{align}\allowdisplaybreaks 
&\Pro(\xi_{\Tf}=1,\cB_{\T_1}^{\T_K})
=\mathop{\sum^{n}\cdots\sum^{n}}_{t>t_K>\cdots>t_1>0}\!\!\!\Pro(\xi_{t}=1,\Tf=t,\T_K=t_K,\ldots,\T_1=t_1,\cB_{t_1}^{t_K})\nonumber\allowdisplaybreaks\\ \allowdisplaybreaks
&~~=\mathop{\sum^{n}\cdots\sum^{n}}_{t>t_K>\cdots>t_1>0}\!\!\!\Exp[\ind{\xi_{t}=1}\ind{\Tf=t}\ind{\T_K=t_k}\cdots\ind{\T_1=t_1}\I_{\cB_{t_1}^{t_K}}]\nonumber\allowdisplaybreaks\\ \allowdisplaybreaks
&~~=\mathop{\sum^{n}\cdots\sum^{n}}_{t>t_K>\cdots>t_1>0}\!\!\!\Exp[\Exp[\ind{\xi_{t}=1}|\cZ_t^K]\ind{\Tf=t}\ind{\T_K=t_k}\cdots\ind{\T_1=t_1}\I_{\cB_{t_1}^{t_K}}]\nonumber\allowdisplaybreaks\\ \allowdisplaybreaks
&~~=\mathop{\sum^{n}\cdots\sum^{n}}_{t>t_K>\cdots>t_1>0}\!\!\!\Exp[\Pro(\xi_{t}=1|\cZ_t^K)\ind{\Tf=t}\ind{\T_K=t_k}\cdots\ind{\T_1=t_1}\I_{\cB_{t_1}^{t_K}}]\nonumber\allowdisplaybreaks\\ \allowdisplaybreaks
&~~=\mathop{\sum^{n}\cdots\sum^{n}}_{t>t_K>\cdots>t_1>0}\!\!\!\Exp\Big[\frac{t}{n}\ind{z_t=1}\ind{\Tf=t}\ind{\T_K=t_k}\cdots\ind{\T_1=t_1}\I_{\cB_{t_1}^{t_K}}\Big]\nonumber\allowdisplaybreaks\\ \allowdisplaybreaks
&~~=\Exp\Big[\frac{\Tf}{n}\ind{z_{\Tf=1}}\I_{\cB_{\T_1}^{\T_{K}}}\Big],\!
\label{eq:mbifla1}\allowdisplaybreaks
\end{align} 
where for the third equality we used the fact that the indicators $\I_{\cB_{t_1}^{t_K}},\ind{\Tf=t},\ind{\T_k=t_k},k=1,\ldots,K$, are measurable with respect to $\cZ_t^K$ and can therefore be placed outside the inner expectation, while for the second last equality we used \eqref{eq:lem2-1}. If we 
substitute \eqref{eq:mbifla1} into \eqref{eq:decomp0} we can rewrite the success probability as
\begin{equation} 
\Psucc=\sum_{k=1}^K\Pro(\xi_{\T_k}=1,\D_{\T_k}=1,\cB_{\T_1}^{\T_{k-1}})+
\Exp\Big[\frac{\Tf}{n}\ind{z_{\Tf}=1}\I_{\cB_{\T_1}^{\T_{K}}}\Big].
\label{eq:decomp}
\end{equation} 

We can now continue with the task of optimizing \eqref{eq:decomp} over all querying times $\T_1,\ldots,\T_K$, the final time $\Tf$ and the decision functions $\D_{\T_1},\ldots,\D_{\T_K}$. We will achieve this goal step-by-step. We start by conditioning on $\{\T_K=t_K,\ldots,\T_1=t_1,\cB_{t_1}^{t_K}\}$ and first optimize over $\Tf>t_K$, followed by a second optimization over $\D_{\T_K}$. This will result in an expression that depends on $\T_1,\ldots,\T_K$ and $\D_{\T_1},\ldots,\D_{\T_{K-1}}$ with a form which will turn out to be similar to \eqref{eq:decomp} but with the sum reduced by one term. Continuing this idea of first conditioning on the previous querying times and the corresponding event $\cB$ we are going to optimize successively over the pairs $(\Tf,\D_{\T_K}),(\T_K,\D_{\T_{K-1}}),\ldots,(\T_2,\D_{\T_1})$ and then, finally, over $\T_1$. The outcome of this sequence of \textit{dependent} optimizations is presented in the next theorem which constitutes our main result. As expected, in this theorem we will identify the optimal version of all search components and also specify the overall optimal success probability.
\vspace{-0.15cm}

\begin{theorem} \label{th:1}
For $t=n,n-1,\ldots,1,$ and $k=K,K-1,\ldots,0,$ define recursively in $t$ and $k$ the deterministic sequences $\{\A_t^k\},\{\U_t^k\}$ by
\begin{align}  \allowdisplaybreaks
\A_{t-1}^k&=\A_t^k\Big(1-\frac{1}{t}\Big)+\max\big\{\U_t^{k+1},\A_t^k\big\}\frac{1}{t},\label{eq:th1-1}\\
\U_t^{k}&=\sum_{m=1}^M\max\big\{\p(m)\frac{t}{n},\q(m)\A_t^k\big\},
\label{eq:th1-2}
\end{align}
initializing with $\A_n^k=0$ and $\U_t^{K+1}=\frac{t}{n}$. Then, for any collection of querying times and final time $\T_1<\cdots<\T_K<\Tf$ and any collection of decision functions $\D_{\T_1},\ldots,\D_{\T_K}$ that conform with Remark\,\ref{rem:3}, if we define for $k=K,K-1,\ldots,0$ the sequence $\{\P_k\}$ by
\begin{equation}
\P_k=\sum_{\ell=1}^{k}\Pro(\xi_{\T_\ell}=1,\D_{\T_\ell}=1,\cB_{\T_1}^{\T_{\ell-1}})+
\Exp\big[\U_{\T_{k+1}}^{k+1}\ind{z_{\T_{k+1}}=1}\I_{\cB_{\T_1}^{\T_{k}}}\big],
\label{eq:th1-2.5}
\end{equation}
we have
\begin{equation}
\Psucc=\P_{K}\leq\P_{K-1}\leq\cdots\leq\P_0\leq\A_0^0.
\label{eq:th1-3}
\end{equation}
The upper bound $\A_0^0$ in \eqref{eq:th1-3} is independent of any search strategy and constitutes the maximal achievable success probability. This optimal performance can be attained if we select the querying times according to
\begin{equation}
\T_{k}=\min\big\{t>\T_{k-1}:\U_t^{k}\ind{z_t=1}\geq \A_t^{k-1}
\big\},
\label{eq:th1-4}
\end{equation}
and the decision functions to satisfy
\begin{equation}
\D_{\T_k}=\bigg\{\begin{array}{cl}
1,&\text{if\/}~\p(\zeta_{\T_k})\frac{\T_k}{n}\geq\q(\zeta_{\T_k})\A_{\T_k}^k\\[4pt]
0,&\text{if\/}~\p(\zeta_{\T_k})\frac{\T_k}{n}<\q(\zeta_{\T_k})\A_{\T_k}^k,
\end{array}
\label{eq:th1-5}
\end{equation}
where $\zeta_{\T_k}$ is the response of the expert at querying time $\T_k$.
\end{theorem}
\begin{proof}
As mentioned, Theorem\,\ref{th:1} constitutes our main result because we identify the optimal version of all the search components and the corresponding maximal success probability. In particular, we have \eqref{eq:th1-4} for the optimal querying times $\T_1,\ldots,\T_K$ and final time $\Tf$ (we recall that $\Tf=\T_{K+1}$), while the optimal version of the decision functions is depicted in \eqref{eq:th1-5}. The complete proof is presented in the Appendix.
\end{proof}

\section{Simplified Form of the Optimal Components}
Theorem\,\ref{th:1} offers explicit formulas for the optimal version of all search components. We recall that in the existing literature, querying and final stopping are defined in terms of very simple rules involving thresholds. For this reason in this section our goal is to develop similar rules for our search strategy. The next lemma identifies certain key monotonicity properties of the sequences introduced in Theorem\,\ref{th:1} that will help us achieve this goal.

\begin{lemma}\label{lem:3}
For fixed $t$ the two sequences $\{\A_t^k\}$ and $\{\U_t^k\}$ are decreasing in $k$, while for fixed $k$ we have
$\{\A_t^k\}$ decreasing and $\{\U_t^k\}$ increasing in $t$. Finally, at the two end points we observe $\A_n^k\leq\U_n^{k+1}$ and $\A_0^k\geq\U_0^{k+1}$.
\end{lemma}
\begin{proof}
With the help of this lemma we will be able to produce simpler versions of the optimal components. We can find the complete proof in the Appendix.
\end{proof}

Let us use the results of Lemma\,\ref{lem:3} to examine \eqref{eq:th1-4} and \eqref{eq:th1-5}. We note that \eqref{eq:th1-4} can be true only if $z_t=1$. Under this assumption and because of the increase of $\{\U_t^{k}\}$ and decrease of $\{\A_t^{k-1}\}$ with respect to $t$, combined with their corresponding values at the two end points $t=0,n$, we understand that there exists a time threshold $r_k$ such that $\U_t^{k}\geq\A_t^{k-1}$ for $t\geq r_k$ while the inequality is reversed for $t<r_k$. The threshold $r_k$ can be identified beforehand by comparing the terms of the two deterministic sequences $\{\U_t^{k}\},\{\A_t^{k-1}\}$. With the help of $r_k$ we can then equivalently write $\T_k$ as follows
\begin{equation}
\T_k=\min\big\{t\geq\max\{\T_{k-1}+1,r_k\}:~z_t=1\big\},
\label{eq:Tk}
\end{equation}
namely, we make the $k$-th query the first time after the previous querying time $\T_{k-1}$, but no sooner than the time threshold $r_k$, we encounter $z_t=1$. A similar conclusion applies to the final time $\Tf$ where the corresponding threshold is $r_{\rm f}=r_{K+1}$. 

Regarding now \eqref{eq:th1-5}, namely the decision whether to stop at the $k$-th querying time or not, again because of the decrease of $\{\A_t^k\}$ (from Lemma\,\ref{lem:3}) and the increase of $\{\frac{t}{n}\}$ with respect to $t$ and also the fact that $\A_n^k=0$ and $\A_0^k>0$, we can conclude that there exist thresholds $s_k(m),m=1,\ldots,M$, that depend on the expert response $\zeta_{\T_k}=m$ so that $\p(m)\frac{t}{n}\geq \q(m)\A_t^k$ for $t\geq s_k(m)$ while the inequality is reversed when $t<s_k(m)$. The precise definition of $s_k(m)$ is
\begin{equation}
s_k(m)=\min\big\{t>0:~\p(m)\frac{t}{n}\geq \q(m)\A_t^k\big\}.
\label{eq:skm}
\end{equation}
With the help of the thresholds $s_k(m)$ which can be computed beforehand, we can equivalently write the optimal decision as
\begin{equation}
\D_{\T_k}=\bigg\{\begin{array}{cl} 1,&\text{if\/}~\T_k\geq s_k(\zeta_{\T_k})\\[4pt]
0,&\text{if\/}~\T_k< s_k(\zeta_{\T_k}),
\end{array} 
\label{eq:Dk}
\end{equation}
where $\zeta_{\T_k}$ is the expert response at querying time $\T_k$.
In other words if we make the $k$-th query at time $\T_k$ and the expert responds with $\zeta_{\T_k}$ then if $\T_k$ is no smaller than the time threshold $s_k(\zeta_{\T_k})$ we terminate the search. If $\T_k$ is strictly smaller than $s_k(\zeta_{\T_k})$ then we continue to the next query (or final-stop if $k=K$). 

At this point we have identified the optimal version of all components of the search strategy. In Table\,\ref{tab:1} we summarize the 
formulas we need to apply in order to compute the corresponding thresholds and also present the way these thresholds must be employed to implement the optimal search strategy.
\begin{table}
\caption{Optimal Search Strategy.}
\label{tab:1}
\begin{tabular}{l}
\toprule
Let $\A_{n}^k=0$ and $\U_t^{K+1}=\frac{t}{n}$.
For $t=n,n-1,\ldots,1,$ and $k=K,\ldots,1,0,$ compute:\\
\midrule
\addlinespace[2pt]
$\A_{t-1}^k=\A_t^k\Big(1-\frac{1}{t}\Big)+\max\Big\{\U_t^{k+1},\A_t^k\Big\}\frac{1}{t}$\\
$\U_t^k=\sum_{m=1}^M\max\{\p(m)\frac{t}{n},\q(m)\A_t^k\}$\\
\addlinespace[2pt]
then $\A_0^0$ is the maximal success probability.\\
\addlinespace[7pt]
\toprule
For 
$k=1,\ldots,K$, find
the thresholds for the querying times and the final time:\\
\midrule
\addlinespace[2pt]
$r_k=\min\{t>0:\U_t^{k}\geq\A_t^{k-1}\}$\\
\addlinespace[2pt]
$r_{\rm f}=\min\{t>0:\frac{t}{n}\geq\A_t^K\}$.\\
\addlinespace[7pt]
\toprule
With $\T_0=0$, the optimal querying times and the final time are defined by:\\
\midrule
\addlinespace[2pt]
$\T_k=\min\{t\geq\max\{\T_{k-1}+1,r_k\}: z_t=1\}$, if we have not terminated at $\T_{k-1}$\\
$\Tf=\min\{t\geq\max\{\T_{K}+1,r_{\rm f}\}: z_t=1\}$, if we have not terminated at $\T_K$.\\
\addlinespace[7pt]
\toprule
For 
$k=1,\ldots,K$ and $m=1,\ldots,M$, find 
the decision thresholds:\\
\midrule
\addlinespace[2pt]
$s_k(m)=\min\{t>0:\p(m)\frac{t}{n}\geq\q(m)\A_t^k\}$.\\
\addlinespace[7pt]
\toprule
The optimal decision whether to terminate at $\T_k$ or not is defined by:\\
\midrule
\addlinespace[2pt]
For an expert response $\zeta_{\T_k}\in\{1,\ldots,M\}$, stop at $\T_k$ if $\T_k\geq s_k(\zeta_{\T_k})$\\
otherwise proceed to the next query (or final time if $k=K$).\\
\addlinespace[2pt]
\bottomrule
\end{tabular}
\end{table}

\begin{remark}\label{rem:4} Even though it is not immediately evident from the previous analysis, the probabilistic description of the querying times, final time and decision functions enjoy a notable stationarity characteristic (also pointed out in \cite{LM} for the infallible expert case). In particular, the form of the final time $\Tf$ is \textit{independent} of the maximal number $K$ of queries. This means that the threshold $r_{\rm f}$ does not depend on $K$ and it is in fact the same as the unique threshold of the classical secretary problem. The same observation applies to any querying time $\T_{K-k}$ and decision function $\D_{\T_{K-k}}$. Their corresponding thresholds $r_{K-k}$ and $s_{K-k}(m)$ do not depend on $K$ but \textit{only} on $k$. This observation basically suggests that if we have identified the optimal components for some maximal value $K$ and we are interested in decreasing $K$ then we do not need to recompute the components. We simply start from the thresholds of the last querying time and decision function and go towards the first and we stop when we have collected the desired number of components. Similarly, if we increase $K$ then we keep the optimal components computed for the original smaller $K$ and add more components in the beginning by applying the formulas of Table\,\ref{tab:1}.
\end{remark}

\begin{remark}\label{rem:5}
Once more, we would like to emphasize that the search strategy presented in Table\,\ref{tab:1} is the optimum under the assumption that we allow \textit{at most one query per object}. We should however point out that when the expert provides faulty answers it clearly makes sense to query more than once per object in order to improve our trust in the expert responses. Unfortunately the corresponding analysis turns out to be significantly more involved compared to our current results, as one can easily confirm by considering the simple example of $K=2$ queries. For this reason, we believe, this more general setting requires separate consideration.
\end{remark}

\begin{remark}\label{rem:6}
Being able to query does not necessarily guarantee a success probability that approaches 1. This limit is attainable only in the case of an infallible expert. Unfortunately, when responses may be wrong, we can improve the success probability but we can only reach a maximal value which is strictly smaller than 1 even if we query with every object (i.e.~$K=n$). To see this fact consider the extreme case where the expert outputs $M=2$ values with uniform probabilities $\p(1)=\p(2)=\q(1)=\q(2)=0.5$. It is clear that responses under this probabilistic model are completely useless for any number of queries. Hence, we expect the resulting optimal scheme to be equivalent to the classical secretary problem (with no queries) which (see \cite{GM}) enjoys a success probability that approximates the value $e^{-1}\approx0.3679$ for large $n$. Under the probabilistic response model, in order to experience success probabilities that approach 1, we \textit{conjecture} that we must allow multiple queries per object and a maximal number $K$ of queries which exceeds the number of objects, that is, $K>n$. Of course, again, we need to exclude the uniform probability model because it continues to be equivalent to the classical secretary problem with no queries even when multiple queries per object are permitted.
\end{remark}

\begin{remark}\label{rem:7}
One of our reviewers suggested a very interesting alternative to optimally decide whether to stop or continue after each query. Without loss of generality we may assume that the likelihood ratios satisfy $\frac{\p(1)}{\q(1)}\geq\frac{\p(2)}{\q(2)}\geq\cdots\geq\frac{\p(M)}{\q(M)}$. Indeed this is always possible by numbering the expert responses according to the rank of their corresponding likelihood ratios. Clearly a larger ratio implies a higher likelihood for the object to be the best. For combinations of $t$ and $k$ let us define the threshold sequence $\{m_t^k\}$
$$
m_t^k=\left\{\begin{array}{ll}\text{arg}\max_m\Big\{\frac{\p(m)}{\q(m)}\geq\A_t^k\frac{n}{t}\Big\}&\text{when the inequality can be satisfied for some}~m\\
0&\text{when the inequality cannot be satisfied for any}~m.
\end{array}\right.
$$
Suppose now that we have followed the optimal strategy and we are at the $k$-th querying time $\T_k$ with the expert responding with $\zeta_{\T_k}$. We can then propose the following alternative termination rule: Stop when $\zeta_{\T_k}\leq m_{\T_k}^k$ and continue to the next query if $\zeta_{\T_k}> m_{\T_k}^k$. Under the assumption of the monotonicity of the likelihood ratios we can show that the two termination rules, namely the one presented here and the optimal depicted in Table\,\ref{tab:1} produce exactly the same decisions regarding stopping after querying. With the help of the monotonicity properties of $\{\A_t^k\}$ established in Lemma\,\ref{lem:3}, we can also demonstrate that the threshold sequence $\{m_t^k\}$ is non-decreasing in $t$ and $k$.
\end{remark}
%

\section{Numerical Example}
Let us now apply the formulas of Table\,\ref{tab:1} to a particular example. We consider the case of $n=100$ objects where the expert outputs $M=2$ values. This means that the random response model contains the probabilities $\p(m),\q(m),m=1,2$. We focus on the symmetric case $\p(1)=\q(2)$, meaning that $\p(1)=1-\p(2)=1-\q(1)=\q(2)$, which can be parametrized with the help of a single parameter $\p=\p(1)=\q(2)$. We assign to $\p$ the following values $\p=0.5,0.6,0.7,0.8,0.9,0.95,0.98,1$ and allow a maximum of $K=10$ queries in order to observe the effectiveness of the optimal scheme. As mentioned, the case $\p=1$ corresponds to the infallible expert, consequently we expect to match the existing results in the literature. We also note that $\p=0.5$ corresponds to the uniform model therefore expert responses contain no useful information and we expect our scheme to be equivalent to the optimal scheme of the classical secretary problem.

\setlength{\tabcolsep}{2.5pt}
\begin{table}[!h]\scriptsize
\centering
\caption{Thresholds and optimal success probability for $n=100$ objects and $K=10$ queries.}
\label{tab:2}
\renewcommand{\arraystretch}{1.2}
\begin{tabular}{|l|l|rrrrrrrrrr|rrrrrrrrrr|l|}
\hline
\multicolumn{1}{|c|}{$\p$}&\multicolumn{1}{c|}{$r_{\rm f}$}&\multicolumn{10}{c|}{$r_1\div r_{10}$}&\multicolumn{10}{c|}{$s_1(m)\div s_{10}(m)$, $m=1$ top, $m=2$ bottom}&\multicolumn{1}{c|}{$\Psucc$}\\
\hline
0.50&38&1&1&1&1&1&1&1&1&1&1&38&38&38&38&38&38&38&38&38&38&0.3710\\
    &  & & & & & & & & & & &38&38&38&38&38&38&38&38&38&38&\\
\hline
0.60&38&27&27&27&27&27&27&27&27&27&29&27&27&27&27&27&27&27&27&27&27&0.3952\\
    &  & & & & & & & & & & &52&52&52&52&52&52&52&52&52&52&\\
\hline
0.70&38&20&20&20&20&20&20&20&20&22&25&20&20&20&20&20&20&20&20&20&19&0.4568\\
    &  & & & & & & & & & & &66&66&66&66&66&66&66&66&66&66&\\
\hline
0.80&38&14&14&14&14&14&14&15&16&18&24&14&14&14&14&14&14&14&14&14&13&0.5548\\
    &  & & & & & & & & & & &78&78&78&78&78&78&78&78&78&78&\\
\hline
0.90&38&8&8&8&8&9&9&10&12&16&23&8&8&8&8&8&8&8&8&7&6&0.7055\\
    &  & & & & & & & & & & &90&90&90&90&90&90&90&90&90&90&\\
\hline
0.95&38&5&5&5&5&6&7&8&11&15&23&5&5&5&5&5&5&5&4&4&3&0.8173\\
    &  & & & & & & & & & & &95&95&95&95&95&95&95&95&95&95&\\
\hline
0.98&38&2&3&3&3&4&5&7&10&15&23&2&2&2&2&2&2&2&2&2&2&0.9095\\
    &  & & & & & & & & & & &98&98&98&98&98&98&98&98&98&98&\\
\hline
1.00&38&1&1&2&2&3&4&6&10&15&23&1&1&1&1&1&1&1&1&1&1&0.9983\\
    &  & & & & & & & & & & &100&100&100&100&100&100&100&100&100&100&\\
\hline
\end{tabular}
\end{table}
\begin{figure}[t!]
\centering
\includegraphics[scale=0.8]{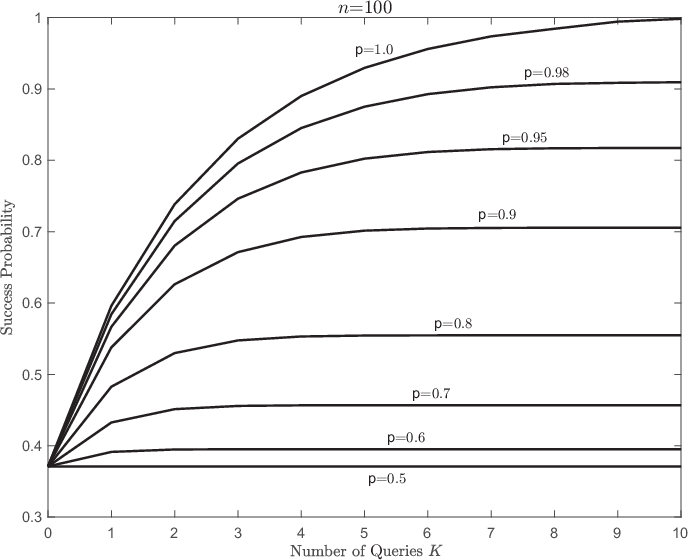}
\caption{Success probability as a function of the number of queries $K$ when $n=100$ objects and $M=2$ responses with symmetric probabilities $\p(1)=1-\p(2)=1-\q(1)=\q(2)=\p$ for $\p=0.5,0.6,0.7,$ $0.8,0.9,0.95,0.98,1$.}
\label{fig:1}
\end{figure}
Using the formulas of Table\,\ref{tab:1} we compute the thresholds $r_{\rm f},r_k,k=1,\ldots,K$ and the decision thresholds $s_k(m),m=1,2,~k=1,\ldots,K$. We can see the corresponding values in Table\,\ref{tab:2} accompanied by the optimal performance delivered by the optimal scheme for $K=10$ queries. In Fig.\,\ref{fig:1} we depict the evolution of the optimal performance for values of $K$ ranging from $K=0$ to $K=10$ where $K=0$ corresponds to the classical secretary problem. Indeed, as we can see from Fig.\,\ref{fig:1} all curves start from the same point which is equal to $\Psucc=0.37104$ namely the maximal success probability in the classical case for $n=100$ (see \cite[Table 2]{GM}). 

In Fig.\,\ref{fig:1} we note the performance of the uniform case $\p=0.5$ which is constant, not changing with the number of queries. As we discussed, this is to be expected since, under the uniform model, expert responses contain no information. It is interesting in this case to compare our optimal scheme to the optimal scheme of the classical version. In the classical case with no queries we recall that the optimal search strategy consists in stopping the first time, but no sooner than $r_{\rm f}=38$, that we observe $z_t=1$. When we allow queries with $\p=0.5$, as we can see, the querying thresholds $r_1\div r_K$ are all equal to 1. This means that the first $K$ times we encounter a relative rank equal to $z_t=1$ we must query. However, stopping at any of the querying times happens only if the querying time is no smaller than $s_k(m)=r_{\rm f}=38$. If all $K$ queries are exhausted before time 38, then we use the terminal time $\T_{\rm f}$ that has a threshold equal to $r_{\rm f}=38$ as well. Consequently final stopping occurs if we encounter a rank equal to 1 no sooner than $r_{\rm f}$. Combining carefully all the possibilities we conclude that we stop at the first time after and including $r_{\rm f}$ that we encounter a rank equal to 1. In other words, we match the classical optimal scheme.

In the last row of Table\,\ref{tab:2} we have the case of an infallible expert. We know that the optimal decision with an infallible expert requires the termination of the search if the expert responds with ``$\{\xi_t=1\}$'' and continuation of the search if the response is ``$\{\xi_t>1\}$''. In our setup, instead, we compare the querying time $\T_k$ to the threshold $s_k(m)$. From the table we see that $s_k(1)=1,~s_k(2)=n$ for all $k=1,\ldots,K$. According to our model $\zeta_{\T_k}=1$ corresponds with certainty (because $\p=1$) to $\xi_{\T_k}=1$, therefore if $\zeta_{\T_k}=1$ we see that we necessarily stop at $\T_k$ since $\T_k\geq s_k(1)=1$. On the other hand, $\zeta_{\T_k}=2$ corresponds with certainty to $\xi_{\T_k}>1$ and when $\zeta_{\T_k}=2$ occurs we can stop only if $\T_k\geq s_k(2)=n$ which is impossible (unless $\T_k=n$ where we necessarily stop since we have exhausted all objects). Therefore, when $\p=1$ our optimal scheme matches the optimal scheme of an infallible expert. This can be further corroborated by comparing our thresholds $r_{\rm f}=38,r_{10}=23$ with the corresponding thresholds $s^*,r^*$ in \cite[Table 3]{GM} and verifying that they are the same with the same success probability (in \cite{GM}, there are tables only for $K=0,1$).

From Fig.\,\ref{fig:1} we can also observe the fact we described in Remark\,\ref{rem:6}, namely that there is an improvement in the success probability, however the optimal value ``saturates'' with the limiting value being strictly less than 1.
An additional conclusion we can draw from this example is that the thresholds also converge to some limiting value. This means that, after some point, increasing $K$ results in repeating the last thresholds and, experiencing the same optimal performance. The only case which does not follow this rule 
and the probability of success converges to 1 as the number of queries increases
is when $\p=1$ which, as mentioned, corresponds to an infallible expert.

A final observation regarding our example is the case where the value of the parameter $\p$ satisfies $\p<0.5$. Using the formulas of Table\,\ref{tab:1} we can show that we obtain exactly the same results as using, instead of $\p$, the value $1-\p>0.5$. The only modification we need to make is to exchange the roles of $m=1$ and $m=2$ in the thresholds $s_k(m)$. We can see why this modification is necessary by considering the extreme case $\p=0$ corresponding to $\Pro(\zeta_t=1|\xi_t=1)=0$. Then, when $M=2$, it is of course true that $\Pro(\zeta_t=2|\xi_t=1)=1$ and, therefore, we now have that the value $\zeta_t=2$ corresponds with certainty to $\{\xi_t=1\}$. This exchange of roles between $m=1$ and $m=2$ continues to apply when $0\leq\p<0.5$.

\section{Acknowledgements}
This work was supported by the US National Science Foundation under Grant CIF 1513373, through Rutgers University, also
in part by the NSF grants NSF CCF 15-26875, The Center for Science of Information at Purdue University, under contract number 239 SBC PURDUE 4101-38050 and by the DARPA molecular informatics program.

We are indepted to the anonymous reviewer whose comments helped us improve considerably the presentation of our results. We would like to particularly mention the decision mechanism proposed by our reviewer (presented in Remark\,\ref{rem:7}) which constitutes a very interesting alternative for optimally deciding whether to terminate or continue after each querying.

\appendix
\section*{Appendix}

\noindent \textbf{Proof of Lemma\,\ref{lem:1}.}~
For any collection of values $\{\xi_t,\ldots,\xi_1\}$ under the uniform model without replacement the validity of \eqref{eq:lem1-1} is well known, since $\xi_\ell$ takes one of $n-\ell+1$ values, each with the same probability $\frac{1}{n-\ell+1}$. Suppose now that the collection $\{\xi_t,\ldots,\xi_1\}$, when occurring sequentially, produces the sequence of ranks $\{z_t,\ldots,z_1\}$ where $1\leq z_t\leq t\leq n$. If we fix a collection $\{z_t,\ldots,z_1\}$ of $t$ ranks, making sure that they conform with the constraint $1\leq z_\ell\leq \ell$ and also select $t$ integers $1\leq i_1<i_2<\cdots<i_t\leq n$ as possible object values, then there is a \textit{unique} way to assign these values to $\{\xi_t,\ldots,\xi_1\}$ in order to produce the specified ranks. Indeed, we start with $\xi_t$ to which we assign the $z_t$-th value from the set $\{i_1,\ldots,i_t\}$, that is, the value $i_{z_t}$. We remove this element from the set of values and then we proceed to $\xi_{t-1}$ to which we assign the $z_{t-1}$-th element from the new list of values, etc. This procedure generates the specified ranks from any subset of $\{1,\ldots,n\}$ of size $t$.

As we just mentioned, for fixed ranks $\{z_t,\ldots,z_1\}$ \textit{any} subset of $t$ integers from the set $\{1,\ldots,n\}$ can be uniquely rearranged and assigned to $\{\xi_t,\ldots,\xi_1\}$ in order to generate the specified rank sequence. There are $\binom{n}{t}$ such possible combinations with each combination having a probability of occurrence equal to $\frac{(n-t)!}{n!}$. Multiplying the two quantities yields the second equality in \eqref{eq:lem1-1}, from which we can then deduce that $\Pro(z_t|\cZ_{t-1})=\frac{\Pro(z_t,\ldots,z_1)}{\Pro(z_{t-1},\ldots,z_1)}=\frac{1}{t}$ and prove the third equality in \eqref{eq:lem1-1}.

Consider now \eqref{eq:lem1-3}. If $\xi_{t_1}=1$ this forces $z_{t_1}=1$ and all ranks for times larger than $t_1$ to be necessarily larger than 1. This is expressed through the indicator $\ind{z_{t_1}=1}\1_{t_1+1}^t=\ind{z_{t_1}=1}(\prod_{\ell=t_1+1}^t\ind{z_\ell>1})$. With $\xi_{t_1}=1$ and $z_{t_1}=1$ let us fix the remaining ranks in $\cZ_t$ assuring they are consistent with the constraint imposed for times larger than $t_1$ and also recalling that $z_\ell$ must take values in the set $\{1,\ldots,\ell\}$. We can now see that we are allowed to select the values of $t-1$ objects from a pool of $n-1$ integers (since the value 1 is already assigned to $\xi_{t_1}$). This generates $\binom{n-1}{t-1}$ combinations and each combination, including also the fact that $\xi_{t_1}=1$, has probability  of occurrence equal to $\frac{(n-t)!}{n!}$. If we multiply the two quantities then we obtain \eqref{eq:lem1-3}. Applying \eqref{eq:lem1-3} for $t_1=t$ (possible since $t\geq t_1$) and using the fact that by definition $\1_{t+1}^t=1$, we obtain \eqref{eq:lem1-2}. This concludes the proof of the lemma.\qed

\vskip0.2cm

\noindent\textbf{Proof of Lemma\,\ref{lem:2}.}~
We begin with \eqref{eq:lem2-1} and we note
\begin{equation}
\Pro(\xi_t=1|\cZ_t^k)=\frac{\Pro(\xi_t=1,\cZ_t^k)}{\Pro(\cZ_t^k)}
=\frac{\Pro(\xi_t=1,\cZ_t^k)}{\Pro(\cZ_t^k)}\ind{z_t=1},
\label{eq:app-lem2-1}
\end{equation}
where the last equality is due to the fact that $\xi_t=1$ forces $z_t$ to become 1 as well, therefore the numerator is 0 if $z_t\neq1$.
This property is captured with the indicator $\ind{z_t=1}$. We can now write
\begin{align*}  \allowdisplaybreaks
&\Pro(\xi_t=1,\cZ_t^k)=\Pro(\xi_t=1,\cZ_{t}^k,\xi_{t_k}=1)+\Pro(\xi_t=1,\cZ_{t}^k,\xi_{t_k}>1)\\
&~~=\Pro(\xi_t=1,\cZ_{t}^k,\xi_{t_k}>1)=\Pro(\xi_t=1,\zeta_{t_k},\cZ_{t}^{k-1},\xi_{t_k}>1)\\
&~~=\Pro(\xi_t=1,\zeta_{t_k},\cZ_{t}^{k-1}|\xi_{t_k}>1)\Pro(\xi_{t_k}>1)\\
&~~=\Pro(\zeta_{t_k}|\xi_{t_k}>1)\Pro(\xi_t=1,\cZ_{t}^{k-1}|\xi_{t_k}>1)\Pro(\xi_{t_k}>1)
=\q(\zeta_{t_k})\Pro(\xi_t=1,\cZ_{t}^{k-1},\xi_{t_k}>1)\\
&~~=\q(\zeta_{t_k})\big\{\Pro(\xi_t=1,\cZ_{t}^{k-1})-\Pro(\xi_t=1,\cZ_{t}^{k-1},\xi_{t_k}=1)\big\}
=\q(\zeta_{t_k})\Pro(\xi_t=1,\cZ_{t}^{k-1}),
\end{align*}
where $\q(\zeta_{t_k})$, following the model in \eqref{eq:rand-model}, denotes the probability of the expert responding with the value $\zeta_{t_k}$ given that $\xi_{t_k}>1$. We also observe that the second and last equality are true due to the fact that fourth is impossible for two objects at different time instances to have the same value, while the forth equality is true because, according to our model when we condition on $\{\xi_{t_k}>1\}$ then $\zeta_{t_k}$ is independent of all ranks, other responses and other object values. 

In order to modify the denominator in \eqref{eq:app-lem2-1}, due to the indicator $\ind{z_t=1}$ in the numerator it is sufficient to analyze the denominator by fixing $z_t=1$. Specifically
%
\begin{align*}  \allowdisplaybreaks
&\Pro(\cZ_{t}^k,z_t=1)=\Pro(\cZ_{t}^k,z_t=1,\xi_{t_k}=1)+\Pro(\cZ_{t}^k,z_t=1,\xi_{t_k}>1)\allowdisplaybreaks\\ \allowdisplaybreaks
&~~=\Pro(\cZ_{t}^k,z_t=1,\xi_{t_k}>1)=\Pro(\zeta_{t_k},\cZ_{t}^{k-1},z_t=1,\xi_{t_k}>1)\allowdisplaybreaks\\ \allowdisplaybreaks
&~~=\Pro(\zeta_{t_k}|\xi_{t_k}>1)\Pro(\cZ_{t}^{k-1},z_t=1,\xi_{t_k}>1)\allowdisplaybreaks\\ \allowdisplaybreaks
&~~=\q(\zeta_{t_k})\big\{\Pro(\cZ_{t}^{k-1},z_t=1)-\Pro(\cZ_{t}^{k-1},z_t=1,\xi_{t_k}=1)\big\}
=\q(\zeta_{t_k})\Pro(z_t=1,\cZ_{t}^{k-1}),
\end{align*}
where in the third and the last equality we used the fact that for $t>t_k$ we cannot have $z_t=1$ when $\xi_{t_k}=1$ because this requires $\xi_t<\xi_{t_k}$ which is impossible since $\xi_{t_k}=1$. Dividing the numerator by the denominator proves that
\begin{align*}
\Pro(\xi_t=1|\cZ_t^k)&=\frac{\Pro(\xi_t=1,\cZ_{t}^{k-1})}{\Pro(z_t=1,\cZ_{t}^{k-1})}\ind{z_t=1}\\
&=\frac{\Pro(\xi_t=1,\cZ_t^{k-1})}{\Pro(\cZ_t^{k-1})}\ind{z_t=1}
=\Pro(\xi_t=1|\cZ_t^{k-1})\ind{z_t=1}.
\end{align*}
In other words, 
given that a new relatively best object appears, the probability that it is the best of all the objects is conditionally independent of the previous expert response.
Applying this equality repeatedly we conclude that
\begin{equation}
\Pro(\xi_t=1|\cZ_t^k)=\Pro(\xi_t=1|\cZ_t^0)\ind{z_t=1}=\Pro(\xi_t=1|\cZ_t)\ind{z_t=1}
\label{eq:lem2-A000}
\end{equation}
namely, the conditional probability is independent of all past expert responses. Combining the second equality in \eqref{eq:lem1-1} with \eqref{eq:lem1-2} we can now show
$$
\Pro(\xi_t=1|\cZ_t)=\frac{\Pro(\xi_t=1,\cZ_t)}{\Pro(\cZ_t)}=\frac{\frac{1}{(t-1)!n}\ind{z_t=1}}{\frac{1}{t!}}=\frac{t}{n}\ind{z_t=1},
$$
which if substituted into \eqref{eq:lem2-A000} proves \eqref{eq:lem2-1} and suggests that the desired conditional probability $\Pro(\xi_t=1|\cZ_t^k)$ depends only on $t$ and $z_t$ and not on $\cZ_{t-1}^k$, namely previous ranks and previous expert responses. 

To show \eqref{eq:lem2-2} we observe that
$$
\Pro(\xi_{t_k}=1|\cZ_{t_k}^k)
=\frac{\Pro(\xi_{t_k}=1,\cZ_{t_k}^k)}{\Pro(\cZ_{t_k}^k)}.
$$
For the numerator using similar steps as before, we can write
\begin{align*}
\Pro(\xi_{t_k}=1,\cZ_{t_k}^k)&=\Pro(\xi_{t_k}=1,\zeta_{t_k},\cZ_{t_k}^{k-1})
=\p(\zeta_{t_k})\Pro(\xi_{t_k}=1,\cZ_{t_k}^{k-1})\\
&=\p(\zeta_{t_k})\Pro(\xi_{t_k}=1|\cZ_{t_k}^{k-1})\Pro(\cZ_{t_k}^{k-1})
=\p(\zeta_{t_k})\frac{t_k}{n}\ind{z_{t_k}=1}\Pro(\cZ_{t_k}^{k-1})
\end{align*}
where for the second equality we first conditioned on $\{\xi_{t_k}=1\}$ and used the fact that $\zeta_{t_k}$ is independent of any other information, while for the last equality we applied \eqref{eq:lem2-1}.

Similarly, for the denominator we have
\begin{align}
\Pro(\cZ_{t_k}^k)&=\Pro(\zeta_{t_k},\cZ_{t_k}^{k-1})=\Pro(\xi_{t_k}=1,\zeta_{t_k},\cZ_{t_k}^{k-1})+\Pro(\xi_{t_k}>1,\zeta_{t_k},\cZ_{t_k}^{k-1})\nonumber\\
&=\p(\zeta_{t_k})\Pro(\xi_{t_k}=1,\cZ_{t_k}^{k-1})+\q(\zeta_{t_k})\Pro(\xi_{t_k}>1,\cZ_{t_k}^{k-1})\nonumber\\
&=\big\{\p(\zeta_{t_k})\Pro(\xi_{t_k}=1|\cZ_{t_k}^{k-1})+\q(\zeta_{t_k})\Pro(\xi_{t_k}>1|\cZ_{t_k}^{k-1})\big\}\Pro(\cZ_{t_k}^{k-1})\nonumber\\
&=\Big\{\p(\zeta_{t_k})\frac{t_k}{n}\ind{z_{t_k}=1}+\q(\zeta_{t_k})\Big(1-\frac{t_k}{n}\ind{z_{t_k}=1}\Big)\Big\}\Pro(\cZ_{t_k}^{k-1})
\label{eq:A1-100}
\end{align}
where for the last equality we applied the same idea we used in the last equality of the numerator.
Dividing the numerator by the denominator and using the fact that in the numerator we have the indicator $\ind{z_{t_k}=1}$, it is easy to verify that we obtain the expression appearing in~\eqref{eq:lem2-2}. 

To prove \eqref{eq:lem2-3} we have
\begin{align*}
\Pro(\zeta_{t_k}|\cZ_{t_k}^{k-1},z_{t_k}=1)&=\frac{\Pro(\zeta_{t_k},\cZ_{t_k}^{k-1},z_{t_k}=1)}{\Pro(\cZ_{t_k}^{k-1},z_{t_k}=1)}=
\frac{\Pro(\cZ_{t_k}^{k},z_{t_k}=1)}{\Pro(\cZ_{t_k}^{k-1},z_{t_k}=1)}\\
&=\p(\zeta_{t_k})\frac{t_k}{n}+\q(\zeta_{t_k})\Big(1-\frac{t_k}{n}\Big),
\end{align*}
where for the last equality we used \eqref{eq:A1-100} and applied it for $z_{t_k}=1$.

Let us now demonstrate \eqref{eq:lem2-4} which is the relationship that distinguishes our random model from the classical infallible expert case. We observe that
$$ 
\Pro(z_t=1|\cZ_{t-1}^k,z_{t_k}=1)
=\frac{\Pro(z_t=1,\cZ_{t-1},z_{t_k}=1,\zeta_{t_k},\ldots,\zeta_{t_1})}{\Pro(\cZ_{t-1},z_{t_k}=1,\zeta_{t_k},\ldots,\zeta_{t_1})}
$$ 
As before we can write for the numerator
\begin{align*}
&\Pro(z_t=1,\cZ_{t-1},z_{t_k}=1,\zeta_{t_k},\ldots,\zeta_{t_1})
=\Pro(z_t=1,\cZ_{t-1},z_{t_k}=1,\zeta_{t_k},\ldots,\zeta_{t_1},\xi_{t_1}>1)\\
&\hskip1cm=\Pro(z_t=1,\cZ_{t-1},z_{t_k}=1,\zeta_{t_k},\ldots,\zeta_{t_2},\xi_{t_1}>1)\q(\zeta_{t_1})\\
&\hskip1cm=\Pro(z_t=1,\cZ_{t-1},z_{t_k}=1,\zeta_{t_k},\ldots,\zeta_{t_2})\q(\zeta_{t_1}).
\end{align*}
For the denominator, due to the constraint $z_{t_k}=1$, we can follow similar steps as in the numerator and show
$$  
\Pro(\cZ_{t-1},z_{t_k}=1,\zeta_{t_k},\ldots,\zeta_{t_1})=
\Pro(\cZ_{t-1},z_{t_k}=1,\zeta_{t_k},\ldots,\zeta_{t_2})\q(\zeta_{t_1}).
$$  
Dividing the numerator by the denominator yields
$$  
\frac{\Pro(z_t=1,\cZ_{t-1},z_{t_k}=1,\zeta_{t_k},\ldots,\zeta_{t_1})}{\Pro(\cZ_{t-1},z_{t_k}=1,\zeta_{t_k},\ldots,\zeta_{t_1})}=
\frac{\Pro(z_t=1,\cZ_{t-1},z_{t_k}=1,\zeta_{t_k},\ldots,\zeta_{t_2})}{\Pro(\cZ_{t-1},z_{t_k}=1,\zeta_{t_k},\ldots,\zeta_{t_2})},
$$  
which suggests that the first ratio does not depend on $\zeta_{t_1}$. Following similar steps we can remove all previous expert responses one-by-one and prove that
$$ 
\Pro(z_t=1|\cZ_{t-1}^k,z_{t_k}=1)
=\frac{\Pro(z_t=1,\cZ_{t-1},z_{t_k}=1,\zeta_{t_k})}{\Pro(\cZ_{t-1},z_{t_k}=1,\zeta_{t_k})},
$$ 
namely, the conditional probability depends only on the most recent expert response. It is possible now to obtain more suitable expressions for the numerator and the denominator. We start with the numerator and apply similar steps as above. Specifically
\begin{align*}
\Pro(z_t=1,\cZ_{t-1},z_{t_k}=1,\zeta_{t_k})
&=\Pro(z_t=1,\cZ_{t-1},z_{t_k}=1,\zeta_{t_k},\xi_{t_k}>1)\\
&=\Pro(z_t=1,\cZ_{t-1},z_{t_k}=1)\q(\zeta_{t_k})=\frac{1}{t!}\q(\zeta_{t_k}),
\end{align*}
where for the last expression we used the second equality in \eqref{eq:lem1-1} after observing that $\{z_t=1,\cZ_{t-1},z_{t_k}=1\}$ is simply $\cZ_t$ with two of its elements fixed to 1.

For the denominator we can similarly write
\begin{align*}
&\Pro(\cZ_{t-1},z_{t_k}=1,\zeta_{t_k})=\Pro(\cZ_{t-1},z_{t_k}=1,\zeta_{t_k},\xi_{t_k}=1)+\Pro(\cZ_{t-1},z_{t_k}=1,\zeta_{t_k},\xi_{t_k}>1)\\
&\hskip1cm=\Pro(\cZ_{t-1},z_{t_k}=1,\xi_{t_k}=1)\p(\zeta_{t_k})+\Pro(\cZ_{t-1},z_{t_k}=1,\xi_{t_k}>1)\q(\zeta_{t_k})\\
&\hskip1cm=\Pro(\cZ_{t-1},z_{t_k}=1,\xi_{t_k}=1)\p(\zeta_{t_k})\\
&\hskip2cm+\big\{\Pro(\cZ_{t-1},z_{t_k}=1)-\Pro(\cZ_{t-1},z_{t_k}=1,\xi_{t_k}=1)\big\}\q(\zeta_{t_k})\\
&\hskip1cm=\Big\{\frac{1}{(t-2)!n}\1_{t_k+1}^{t-1}\Big\}\p(\zeta_{t_k})+\Big\{\frac{1}{(t-1)!}-\frac{1}{(t-2)!n}\1_{t_k+1}^{t-1}\Big\}\q(\zeta_{t_k}),
\end{align*}
where to obtain the last expression we applied the second equality of \eqref{eq:lem1-1} combined with \eqref{eq:lem1-3} after observing that by fixing $z_{t_k}=1$ the product $\ind{z_{t_k}=1}\1_{t_k+1}^{t-1}$ produced by \eqref{eq:lem1-3} becomes $\1_{t_k+1}^{t-1}$. Dividing the numerator by the denominator we can verify that the resulting ratio matches the right hand side of \eqref{eq:lem2-4} for the two possible values of $\1_{t_k+1}^{t-1}$, namely 0 or 1. This completes the proof of Lemma\,\ref{lem:2}.\qed
\vskip0.2cm

\noindent\textbf{Proof of Theorem\,\ref{th:1}.}~
We first note that $\Psucc=\P_{K}$ where $\P_{K}$ satisfies \eqref{eq:th1-2.5} for $k=K$ and $\U_t^{K+1}=\frac{t}{n}$. Consider now the general form of $\P_k$ defined in \eqref{eq:th1-2.5}. We focus on the last term which we intend to optimize with respect to $\T_{k+1}$. Observe that
\begin{align}
&\Exp\big[\U_{\T_{k+1}}^{k+1}\ind{z_{\T_{k+1}}=1}\I_{\cB_{\T_1}^{\T_{k}}}\big]\nonumber\\
&~~=\mathop{\sum^{n}\cdots\sum^{n}}_{t_k>\cdots>t_1>0}\Exp\big[\U_{\T_{k+1}}^{k+1}\ind{z_{\T_{k+1}}=1}\ind{\T_{k+1}>t_k}\ind{\T_k=t_k}\cdots\ind{\T_1=t_1}\I_{\cB_{t_1}^{t_{k}}}\big]\nonumber\\
&~~=\mathop{\sum^{n}\cdots\sum^{n}}_{t_k>\cdots>t_1>0}\Exp\Big[\Exp\big[\U_{\T_{k+1}}^{k+1}\ind{z_{\T_{k+1}}=1}|\cZ_{t_k}^k\big]\ind{\T_{k+1}>t_k}\ind{\T_k=t_k}\cdots\ind{\T_1=t_1}\I_{\cB_{t_1}^{t_{k}}}\Big],
\label{eq:th1-A1}
\end{align}
where in the last equality, as we point out in Remark\,\ref{rem:3}, the indicators $\ind{\T_{k+1}>t_k}$, $\ind{\T_k=t_k},\ldots,$ $\ind{\T_1=t_1}$, $\I_{\cB_{t_1}^{t_{k}}}$ are measurable with respect to $\cZ_{t_k}^k$ and, consequently, can be placed outside the inner expectation.

We could isolate the inner expectation and optimize it by solving the optimal stopping problem
\begin{equation} 
\max_{\T_{k+1}>t_{k}}\Exp\big[\U_{\T_{k+1}}^{k+1}\ind{z_{\T_{k+1}}=1}|\cZ_{t_k}^k\big],
\label{eq:th1-A1.25}
\end{equation} 
with respect to $\T_{k+1}$. Unfortunately the proposed optimization turns out to be unnecessarily involved resulting in an optimal reward which is a complicated expression of the information $\cZ_{t_k}^k$. After careful examination, and recalling from \eqref{eq:BandD} that $\I_{\cB_{t_1}^{t_{k}}}$ contains the indicator $\ind{z_{t_k}=1}$, it is sufficient to consider the case where $z_{t_k}$ is fixed to the value 1. This constraint simplifies considerably our analysis and it is the main reason we have developed equalities \eqref{eq:lem2-3}, \eqref{eq:lem2-4} in Lemma\,\ref{lem:2}. We also recall that $z_{t_k}=1$, according to Remark\,\ref{rem:1}, is a prerequisite for querying at $t_k$. 

After this observation, we replace \eqref{eq:th1-A1} with the alternative relationship
\begin{multline}
\Exp\big[\U_{\T_{k+1}}^{k+1}\ind{z_{\T_{k+1}}=1}\I_{\cB_{\T_1}^{\T_{k}}}\big]=\\
\!\!\!\mathop{\sum^{n}\cdots\sum^{n}}_{t_k>\cdots>t_1>0}\!\!\Exp\Big[\Exp\big[\U_{\T_{k+1}}^{k+1}\ind{z_{\T_{k+1}}=1}|\cZ_{t_k}^k,z_{t_k}=1\big]\ind{\T_{k+1}>t_k}\ind{\T_k=t_k}\cdots\ind{\T_1=t_1}\I_{\cB_{t_1}^{t_{k}}}\Big].\!\!
\label{eq:th1-A1.5}
\end{multline}
Again, we emphasize that we are allowed to make this specific conditioning because the value $z_{t_k}=1$ is imposed by the indicator $\ind{z_{t_k}=1}$ contained in $\I_{\cB_{t_1}^{t_{k}}}$. Let us now isolate the inner expectation in \eqref{eq:th1-A1.5} and consider the following optimal stopping problem in place of \eqref{eq:th1-A1.25}
\begin{equation} 
\max_{\T_{k+1}>t_{k}}\Exp\big[\U_{\T_{k+1}}^{k+1}\ind{z_{\T_{k+1}}=1}|\cZ_{t_k}^k,z_{t_k}=1\big].
\label{eq:th1-A2}
\end{equation} 
Following \cite{PS,S2}, for $t>t_k$ we need to define the sequence of optimal rewards $\{\R_t^k\}$ where
\begin{equation} 
\R_t^k=\max_{\T_{k+1}\geq t}\Exp[\U_{\T_{k+1}}^{k+1}\ind{z_{\T_{k+1}}=1}|\cZ_t^k,z_{t_k}=1].
\label{eq:th1-A3}
\end{equation} 
From optimal stopping theory we have that $\{\R_t^{k}\}$ satisfies the backward recursion 
\begin{equation} 
\R_t^k=\max\big\{\U_t^{k+1}\ind{z_t=1},\Exp[\R_{t+1}^k|\cZ_t^k,z_{t_k}=1]\big\},
\label{eq:th1-A4}
\end{equation} 
which must be applied for $t=n,n-1,\ldots,t_k+1$ and initialized with $\R_{n+1}^k=0$. We recall that $t_k$ is excluded from the possible values of $\T_{k+1}$ since we require $\T_{k+1}>t_k$. 

In order to find an explicit formula for the reward, we use the definition of the sequence $\{\A_t^k\}$ from \eqref{eq:th1-1} and we introduce a second sequence $\{\B_{t}^k(m)\}$ satisfying the following backward recursion
\begin{multline}
\B_{t-1}^k(m)=\B_{t}^k(m)\Big(1-\frac{1}{t}\Big)\\
+\big(\A_t^k\!+\!\B_{t}^k(m)-\max\{\U_t^{k+1},\A_t^k\}\big)\frac{1}{t}\frac{\big(\p(m)-\q(m)\big)(t-1)}{\p(m)(t-1)+\q(m)(n-t+1)}
\label{eq:th1-A5}
\end{multline}
$t=n,\ldots,t_k,~m=1,\ldots,M$, which is initialized with $\B_{n}^k(m)=0$.
Actually, we are interested in the expected reward $\V_{t}^k=\Exp[\R_{t+1}^k|\cZ_t^k,z_{t_k}=1]$ for which we intend to show, using (backward) induction, that
\begin{equation}
\V_{t}^k=\A_t^k+\B_{t}^k(\zeta_{t_k})\1_{t_{k}+1}^t.
\label{eq:th1-A6}
\end{equation}
Indeed, we have that \eqref{eq:th1-A6} is true for $t=n$ since both the right and left hand sides are 0. Assume our claim is true for $t$, then we will show that it is also valid for $t-1$. Using \eqref{eq:th1-A4} we can write
\begin{align*}
\R_{t}^k&=\max\{\U_t^{k+1},\V_{t}^k\}\ind{z_t=1}+\V_{t}^k\ind{z_t>1}\\
&=\max\{\U_t^{k+1},\A_t^k+\B_{t}^k(\zeta_{t_k})\1_{t_{k}+1}^t\}\ind{z_t=1}+\big(\A_t^k+\B_{t}^k(\zeta_{t_k})\1_{t_{k}+1}^t\big)\ind{z_t>1}\\
&=\max\{\U_t^{k+1},\A_t^k\}\ind{z_t=1}+\big(\A_t^k+\B_{t}^k(\zeta_{t_k})\1_{t_{k}+1}^{t-1}\big)\ind{z_t>1}.
\end{align*} 
Taking expectations on both sides conditioned on $\{\cZ_{t-1}^k,z_{t_k}=1\}$, using \eqref{eq:lem2-4} and rearranging terms, it is not complicated to verify that $\V_{t-1}^k$ is also equal to $\A_{t-1}^k+\B_{t-1}^k(\zeta_{t_k})\1_{t_{k}+1}^{t-1}$, provided that $\{\A_t^k\}$ and $\{\B_{t}^k(m)\}$ are defined by \eqref{eq:th1-1} and \eqref{eq:th1-A5}, respectively.

Let us now return to the optimization problem in \eqref{eq:th1-A2}. According to our analysis, the optimal reward satisfies
\begin{equation} 
\max_{\T_{k+1}>t_k}\!\!\!\Exp[\U_{\T_{k+1}}^{k+1}\ind{z_{\T_{k+1}}=1}|\cZ_{t_k}^k,z_{t_k}=1]\!=\!\V_{t_k}^k\!=\!
\A_{t_k}^k\!+\!\B_{t_k}^k(\zeta_{t_k})\1_{t_{k}+1}^{t_{k}}\!=\!\A_{t_{k}}^k\!+\!\B_{t_{k}}^k(\zeta_{t_k}),
\label{eq:th1-A7}
\end{equation} 
since, according to our definition, $\1_a^b=1$ when $a>b$.

The next step consists in finding a more convenient expression for the sum $\A_t^k+\B_{t}^k(m)$. Again, using backward induction we prove that
\begin{equation}
\B_{t}^k(m)=\A_t^k\frac{\big(\q(m)-\p(m)\big)t}{\p(m)t+\q(m)(n-t)}.
\label{eq:th1-A6-2}
\end{equation}
Clearly, for $t=n$ this expression is true since both sides are 0. We assume it is true for $t$ and we will show that it is valid for $t-1$. Indeed, if we substitute \eqref{eq:th1-A6-2} into the definition in \eqref{eq:th1-A5} then, after some straightforward manipulations, we end up with the equality
$$
\B_{t-1}^k(m)=\Big\{\A_t^k\Big(1-\frac{1}{t}\Big)+\max\{\U_t^{k+1},\A_t^k\}\frac{1}{t}\Big\}\frac{\big(\q(m)-\p(m)\big)(t-1)}{\p(m)(t-1)+\q(m)(n-t+1)},
$$
which, with the help of \eqref{eq:th1-1}, proves the induction. 
Substituting \eqref{eq:th1-A6-2} into \eqref{eq:th1-A7} provides a more concise expression for the optimal reward 
\begin{equation} %
\max_{\T_{k+1}>t_k}\Exp[\U_{\T_{k+1}}^{k+1}\ind{z_{\T_{k+1}}=1}|\cZ_{t_k}^k,z_{t_k}=1]=
\A_{t_k}^k\frac{\q(\zeta_{t_k})n}{\p(\zeta_{t_k})t_k+\q(\zeta_{t_k})(n-t_k)},
\label{eq:th1-A8}
\end{equation} 
which depends only on the most recent expert response $\zeta_{t_k}$. Using \eqref{eq:th1-A8} we obtain the following (attainable) upper bound for \eqref{eq:th1-A1}
\begin{multline}
\Exp\big[\U_{\T_{k+1}}^{k+1}\ind{z_{\T_{k+1}}=1}\I_{\cB_{\T_1}^{\T_{k}}}\big]
\leq\Exp\Big[\A_{\T_k}^k\frac{\q(\zeta_{\T_k})n}{\p(\zeta_{\T_k})\T_k+\q(\zeta_{\T_k})(n-\T_k)}\I_{\cB_{\T_1}^{\T_{k}}}\Big]\\
=\Exp\Big[\A_{\T_k}^k\frac{\q(\zeta_{\T_k})n}{\p(\zeta_{\T_k})\T_k+\q(\zeta_{\T_k})(n-\T_k)}\ind{z_{\T_k}=1}\ind{\D_{\T_k}=0}\I_{\cB_{\T_1}^{\T_{k-1}}}\Big].
\label{eq:th1-A9}
\end{multline}

The optimal performance, according to optimal stopping theory, can be achieved by the following stopping time
$$
\T_{k+1}=\min\{t>t_k:\U_t^{k+1}\ind{z_t=1}\geq \A_t^k+\B_{t}^k(\zeta_{t_k})\1_{t_{k}+1}^t\}.
$$
The previous stopping rule gives the impression that the optimal $\T_{k+1}$ depends on the expert response value $\zeta_{t_k}$. However, we observe that the only way we can stop is if $z_t=1$ which forces the indicator $\1_{t_{k}+1}^t$ to become 0. Consequently, the optimal version of $\T_{k+1}$ is equivalent to
$$
\T_{k+1}=\min\{t>t_k:\U_t^{k+1}\ind{z_t=1}\geq \A_t^k\},
$$
which is independent of $\zeta_{t_k}$ and proves \eqref{eq:th1-4}.

We conclude that the solution of the optimization problem introduced in \eqref{eq:th1-A2} resulted in the identification of the optimal querying times $\T_1,\ldots,\T_K$ and the optimal final time $\Tf$ (since $\Tf=\T_{K+1}$). 
Let us now see how we can optimize the remaining elements of our search strategy, namely, the decision functions $\D_{\T_1},\ldots,\D_{\T_K}$. Consider the last component of the sum in \eqref{eq:th1-2.5} which can be written as follows
\begin{align}\allowdisplaybreaks  
&\Pro(\xi_{\T_k}=1,\D_{\T_k}=1,\cB_{\T_1}^{\T_{k-1}})\nonumber\allowdisplaybreaks\\ \allowdisplaybreaks
&~~=\mathop{\sum^{n}\cdots\sum^{n}}_{t_k>\cdots>t_1>0}\Exp[\ind{\xi_{t_k}=1}\ind{\T_k=t_k}\cdots\ind{\T_1=t_1}\ind{\D_{t_k}=1}\I_{\cB_{t_1}^{t_{k-1}}}]\nonumber\allowdisplaybreaks\\ \allowdisplaybreaks
&~~=\mathop{\sum^{n}\cdots\sum^{n}}_{t_k>\cdots>t_1>0}\Exp\big[\Pro(\xi_{t_k}=1|\cZ_{t_k}^k)\ind{\T_k=t_k}\cdots\ind{\T_1=t_1}\ind{\D_{t_k}=1}\I_{\cB_{t_1}^{t_{k-1}}}\big]\nonumber\allowdisplaybreaks\\ \allowdisplaybreaks
&~~=\mathop{\sum^{n}\cdots\sum^{n}}_{t_k>\cdots>t_1>0}\Exp\Big[\frac{\p(\zeta_{t_k})t_k}{\p(\zeta_{t_k})t_k+\q(\zeta_{t_k})(n-t_k)}\ind{z_{t_k}=1}\ind{\T_k=t_k}\cdots\ind{\T_1=t_1}\ind{\D_{t_k}=1}\I_{\cB_{t_1}^{t_{k-1}}}\Big]\nonumber\allowdisplaybreaks\\ \allowdisplaybreaks
&~~=\Exp\Big[\frac{\p(\zeta_{\T_k})\T_k}{\p(\zeta_{\T_k})\T_k+\q(\zeta_{\T_k})(n-\T_k)}\ind{z_{\T_k}=1}\ind{\D_{\T_k}=1}\I_{\cB_{\T_1}^{\T_{k-1}}}\Big].
\label{eq:th1-A11}\allowdisplaybreaks
\end{align}  
The second equality is true because we condition on $\cZ_t^k$ and since all indicator functions are measurable with respect to this sigma algebra they can be placed outside the inner expectation which gives rise to the conditional probability. For the third equality we simply apply \eqref{eq:lem2-2}. If we now add the two parts analyzed in \eqref{eq:th1-A9} and \eqref{eq:th1-A11}, we can optimize the sum with respect to $\D_{\T_k}$. In particular
\begin{align}
&\Pro(\xi_{\T_k}=1,\D_{\T_k}=1,\cB_{\T_1}^{\T_{k-1}})+\Exp\big[\U_{\T_{k+1}}^{k+1}\ind{z_{\T_{k+1}}=1}\I_{\cB_{\T_1}^{\T_{k}}}\big]\nonumber\\
&\hskip1cm\leq\Exp\Big[\frac{\p(\zeta_{\T_k})\T_k}{\p(\zeta_{\T_k})\T_k+\q(\zeta_{\T_k})(n-\T_k)}\ind{z_{\T_k}=1}\ind{\D_{\T_k}=1}\I_{\cB_{\T_1}^{\T_{k-1}}}\Big]\nonumber\\
&\hskip2cm+\Exp\Big[\A_{\T_k}^k\frac{\q(\zeta_{\T_k})n}{\p(\zeta_{\T_k})\T_k+\q(\zeta_{\T_k})(n-\T_k)}\ind{z_{\T_k}=1}\ind{\D_{\T_k}=0}\I_{\cB_{\T_1}^{\T_{k-1}}}\Big]\nonumber\\
&\hskip1cm\leq\Exp\Big[\frac{\max\{\p(\zeta_{\T_k})\frac{\T_k}{n},\q(\zeta_{\T_k})\A_{\T_k}^k\}}{\p(\zeta_{\T_k})\frac{\T_k}{n}+\q(\zeta_{\T_k})(1-\frac{\T_k}{n})}\ind{z_{\T_k}=1}\I_{\cB_{\T_1}^{\T_{k-1}}}\Big].
\label{eq:th1-A12}
\end{align}
We attain the last upper bound if we select $\D_{\T_k}=1$ (i.e.~stop) when $\p(\zeta_{\T_k})\frac{\T_k}{n}\geq\q(\zeta_{\T_k})\A_{\T_k}^k$ and $\D_{\T_k}=0$ (i.e.~continue to the next query or final time if $k=K$) when the inequality is reversed. This clearly establishes \eqref{eq:th1-5} and identifies the optimal version of the decision functions.

As we can see the upper bound in \eqref{eq:th1-A12} is written in terms of the expert response $\zeta_{\T_k}$. In order to obtain an expression which has the same form as the one in \eqref{eq:th1-2.5} we need to average out this random variable. We note
\begin{align}
&\Exp\Big[\frac{\max\{\p(\zeta_{\T_k})\frac{\T_k}{n},\q(\zeta_{\T_k})\A_{\T_k}^k\}}{\p(\zeta_{\T_k})\frac{\T_k}{n}+\q(\zeta_{\T_k})(1-\frac{\T_k}{n})}\ind{z_{\T_k}=1}\I_{\cB_{\T_1}^{\T_{k-1}}}\Big]\nonumber\\
&~=\!\!\!\mathop{\sum^{n}\cdots\sum^{n}}_{t_k>\cdots>t_1>0}\!\Exp\Big[\frac{\max\{\p(\zeta_{t_k})\frac{t_k}{n},\q(\zeta_{t_k})\A_{t_k}^k\}}{\p(\zeta_{t_k})\frac{t_k}{n}+\q(\zeta_{t_k})(1-\frac{t_k}{n})}\ind{z_{t_k}=1}\ind{\T_k=t_k}\cdots\ind{\T_1=t_1}\I_{\cB_{t_1}^{t_{k-1}}}\Big]\nonumber\\
&~=\!\!\!\mathop{\sum^{n}\cdots\sum^{n}}_{t_k>\cdots>t_1>0}\!\sum_{m=1}^M\!\Exp\Big[\frac{\max\{\p(m)\frac{t_k}{n},\q(m)\A_{t_k}^k\}}{\p(m)\frac{t_k}{n}+\q(m)(1-\frac{t_k}{n})}\ind{\zeta_{t_k}=m}\ind{z_{t_k}=1}\ind{\T_k=t_k}\cdots\ind{\T_1=t_1}\I_{\cB_{t_1}^{t_{k-1}}}\Big]\nonumber\displaybreak\\
&~=\!\!\!\mathop{\sum^{n}\cdots\sum^{n}}_{t_k>\cdots>t_1>0}\!\sum_{m=1}^M\!\frac{\max\{\p(m)\frac{t_k}{n},\q(m)\A_{t_k}^k\}}{\p(m)\frac{t_k}{n}+\q(m)(1-\frac{t_k}{n})}\Exp[\ind{\zeta_{t_k}=m}\ind{z_{t_k}=1}\ind{\T_k=t_k}\cdots\ind{\T_1=t_1}\I_{\cB_{t_1}^{t_{k-1}}}],\!\!\!\!
\label{eq:th1-A13}
\end{align}
with the last equality being true because the ratio is deterministic. Consider the last expectation separately. Because of the existence of the indicator $\ind{z_{t_k}=1}$ we can write
\begin{align*}
&\Exp[\ind{\zeta_{t_k}=m}\ind{z_{t_k}=1}\ind{\T_k=t_k}\cdots\ind{\T_1=t_1}\I_{\cB_{t_1}^{t_{k-1}}}]\\
&\hskip1cm=\Exp[\Pro(\zeta_{t_k}=m|\cZ_{t_k}^{k-1},z_{t_k}=1)\ind{z_{t_k}=1}\ind{\T_k=t_k}\cdots\ind{\T_1=t_1}\I_{\cB_{t_1}^{t_{k-1}}}]\\
&\hskip1cm=\Exp\Big[\Big\{\p(m)\frac{t_k}{n}+\q(m)\Big(1-\frac{t_k}{n}\Big)\Big\} \ind{z_{t_k}=1}\ind{\T_k=t_k}\cdots\ind{\T_1=t_1}\I_{\cB_{t_1}^{t_{k-1}}}\Big]\\
&\hskip1cm=\Big\{\p(m)\frac{t_k}{n}+\q(m)\Big(1-\frac{t_k}{n}\Big)\Big\} \Exp[\ind{z_{t_k}=1}\ind{\T_k=t_k}\cdots\ind{\T_1=t_1}\I_{\cB_{t_1}^{t_{k-1}}}].
\end{align*}
According to Remark\,\ref{rem:3} all indicators are $\{\cZ_{t_k}^{k-1},z_{t_k}=1\}$-measurable and this allowed us in the first equation to position them outside the inner expectation which resulted in the conditional probability. For the second equation we used \eqref{eq:lem2-3}. Substituting into \eqref{eq:th1-A13} we obtain
\begin{align}
&\Exp\Big[\frac{\max\{\p(\zeta_{\T_k})\frac{\T_k}{n},\q(\zeta_{\T_k})\A_{\T_k}^k\}}{\p(\zeta_{\T_k})\frac{\T_k}{n}+\q(\zeta_{\T_k})(1-\frac{\T_k}{n})}\ind{z_{\T_k}=1}\I_{\cB_{\T_1}^{\T_{k-1}}}\Big]\nonumber\\
&~~=\!\!\mathop{\sum^{n}\cdots\sum^{n}}_{t_k>\cdots>t_1>0}\!\Big(\sum_{m=1}^M\max\Big\{\p(m)\frac{t_k}{n},\q(m)\A_{t_k}^k\Big\}\Big)
\Exp[\ind{z_{t_k}=1}\ind{\T_k=t_k}\cdots\ind{\T_1=t_1}\I_{\cB_{t_1}^{t_{k-1}}}]\nonumber\\
&~~=\!\!\mathop{\sum^{n}\cdots\sum^{n}}_{t_k>\cdots>t_1>0}\U_{t_k}^k\!
\Exp[\ind{z_{t_k}=1}\ind{\T_k=t_k}\cdots\ind{\T_1=t_1}\I_{\cB_{t_1}^{t_{k-1}}}]\nonumber\\
&~~=\!\!\mathop{\sum^{n}\cdots\sum^{n}}_{t_k>\cdots>t_1>0}\!
\Exp[\U_{t_k}^k\ind{z_{t_k}=1}\ind{\T_k=t_k}\cdots\ind{\T_1=t_1}\I_{\cB_{t_1}^{t_{k-1}}}]
\!=\!\Exp\big[\U_{\T_k}^k\ind{z_{\T_k}=1}\I_{\cB_{\T_1}^{\T_{k-1}}}],
\label{eq:th1-A14}
\end{align}
where we recall that $\U_t^k$ is deterministic and defined in \eqref{eq:th1-2}.

As we have seen, the sum of the two terms in \eqref{eq:th1-A12} is optimized in \eqref{eq:th1-A14}. A~direct consequence of this optimization is the following inequality
\begin{align*}
\P_k&=\sum_{\ell=1}^{k}\Pro(\xi_{\T_\ell}=1,\D_{\T_\ell}=1,\cB_{\T_1}^{\T_{\ell-1}})+
\Exp\big[\U_{\T_{k+1}}^{k+1}\ind{z_{\T_{k+1}}=1}\I_{\cB_{\T_1}^{\T_{k}}}\big]\\
&=\sum_{\ell=1}^{k-1}\Pro(\xi_{\T_\ell}=1,\D_{\T_\ell}=1,\cB_{\T_1}^{\T_{\ell-1}})\\
&\hskip2cm +\Pro(\xi_{\T_k}=1,\D_{\T_k}=1,\cB_{\T_1}^{\T_{k-1}})+\Exp\big[\U_{\T_{k+1}}^{k+1}\ind{z_{\T_{k+1}}=1}\I_{\cB_{\T_1}^{\T_{k}}}\big]\\
&\leq\sum_{\ell=1}^{k-1}\Pro(\xi_{\T_\ell}=1,\D_{\T_\ell}=1,\cB_{\T_1}^{\T_{\ell-1}})+
\Exp\big[\U_{\T_{k}}^{k}\ind{z_{\T_{k}}=1}\I_{\cB_{\T_1}^{\T_{k-1}}}\big]=\P_{k-1},
\end{align*}
namely $\P_k\leq\P_{k-1}$. Repeated application of this fact for $k=K,\ldots,1$, proves \eqref{eq:th1-3} except for the last inequality. In other words, we have
$$
\Psucc=\P_K\leq\P_{K-1}\leq\cdots\leq\P_0=\Exp[\U_{\T_1}^1].
$$
The last expectation can be further optimized with respect to $\T_1$ using our results from \eqref{eq:th1-A3} for $k=0$. In fact the corresponding optimization is far simpler than the general case considered in \eqref{eq:th1-A3} because there is no query response available and therefore the elements of the sequences $\{\B_t^0(m)\}$ are equal to 0. This also implies that the corresponding optimal average reward, from \eqref{eq:th1-A6}, is equal to $\A_t^0$ which establishes the last inequality in \eqref{eq:th1-3} and concludes the proof of our main theorem.\qed 
\vskip0.2cm

\noindent\textbf{Proof of Lemma\,\ref{lem:3}.}~
Let us first prove $\A_t^k\leq\A_t^{k-1}$. We will show this fact using backward induction. To show its validity for $k=K$ we note from \eqref{eq:th1-2} that
$$
\U_t^K=\sum_{m=1}^M\max\big\{\p(m)\frac{t}{n},\q(m)\A_t^K\big\}\geq
\sum_{m=1}^M\p(m)\frac{t}{n}=\frac{t}{n}=\U_t^{K+1}.
$$
Applying now \eqref{eq:th1-1} for $k=K$ and $k=K-1$, using the previous inequality and the fact that $\A_n^K=\A_n^{K-1}=0$, we can easily show using backward induction in $t$ that $\A_t^{K}\leq\A_t^{K-1}$. Suppose now it is true for $k$, that is, $\A_t^k\leq\A_t^{k-1}$, then we will show that $\A_t^{k-1}\leq\A_t^{k-2}$. From $\A_t^k\leq\A_t^{k-1}$ and \eqref{eq:th1-2} we conclude that $\U_t^k\leq\U_t^{k-1}$. Expressing $\A_t^{k-1}$ and $\A_t^{k-2}$ with the help of \eqref{eq:th1-1}, using the facts that $\U_t^k\leq\U_t^{k-1}$ and $\A_n^{k-1}=\A_n^{k-2}=0$, we can again prove using backward induction in $t$ that $\A_t^{k-1}\leq\A_t^{k-2}$ therefore completing the induction. The monotonicity in $k$ of $\U_t^k$ is a direct consequence of \eqref{eq:th1-2} and of the same monotonicity of $\A_t^k$.

To establish that $\{A_t^k\}$ is decreasing in $t$ we use \eqref{eq:th1-1} and observe that
$$
\A_{t-1}^k-\A_t^k=\Big(\max\Big\{\U_t^{k+1},\A_t^k\Big\}-\A_t^k\Big)\frac{1}{t}\geq0
$$
which proves the desired inequality. Demonstrating that $\{\U_t^k\}$ is increasing in $t$ requires more work.
From the definition in \eqref{eq:th1-2} and using \eqref{eq:th1-1} to replace $\A_{t-1}^k$, we have that
\begin{align*}
\U_{t-1}^{k}&=\sum_{m=1}^M\max\Big\{\p(m)\frac{t-1}{n},\q(m)\A_{t-1}^k\Big\}\\
&=\sum_{m=1}^M\max\Big\{\p(m)\frac{t}{n}\Big(1-\frac{1}{t}\Big),\q(m)\Big[\A_t^k\Big(1-\frac{1}{t}\Big)+\max\{\U_t^{k+1},\A_t^k\}\frac{1}{t}\Big]\Big\}\\
&\leq\sum_{m=1}^M\left(\max\Big\{\p(m)\frac{t}{n},\q(m)\A_t^k\Big\}\Big(1-\frac{1}{t}\Big)+\q(m)\max\{\U_t^{k+1},\A_t^k\}\frac{1}{t}\right)\\
&=\U_t^k\Big(1-\frac{1}{t}\Big)+\max\big\{\U_t^{k+1},\A_t^k\big\}\frac{1}{t}=\U_t^k+\Big(\max\big\{\U_t^{k+1},\A_t^k\big\}-\U_t^k\Big)\frac{1}{t},
\end{align*}
with the inequality being true because $\max\{ad,bd+c\}\leq\max\{a,b\}d+c$ when $c,d\geq0$.
To establish that $\U_{t-1}^k\leq\U_t^k$ it suffices to prove that $\U_t^k\geq\max\{\U_t^{k+1},\A_t^k\}$, namely that $\U_t^k\geq\U_t^{k+1}$ (which we already know to be the case) and $\U_t^k\geq\A_t^k$. To show the latter, from its definition in \eqref{eq:th1-2} we can see that
$\U_t^k\geq\sum_{m=1}^M\q(m)\A_t^k=\A_t^k$ and this establishes the desired result. 

To complete our proof we still need to show that $\A_n^k\leq\U_n^{k+1}$ and $\A_0^k\geq\U_0^{k+1}$. We recall that $\A_n^k=0$. On the other hand from \eqref{eq:th1-2} we can see that $\U_n^{k+1}=1$, therefore the first inequality is true. For the second, again from \eqref{eq:th1-2}, we observe that $\U_0^{k+1}=\A_0^{k+1}$ and since we previously established that $\A_t^k$ is decreasing in $k$ for fixed $t$, this proves the second inequality and concludes our proof.\qed

%

\end{document}